\newif\iflong
\newif\ifshort

\longtrue

\iflong
\else
\shorttrue
\fi

\documentclass{article}
\usepackage{ijcai23}

\usepackage{times}
\usepackage{soul}
\usepackage{url}
\usepackage[hidelinks]{hyperref}
\usepackage[utf8]{inputenc}
\usepackage[small]{caption}
\usepackage{graphicx}
\usepackage{amsmath}
\usepackage{amsthm}
\usepackage{booktabs}
\usepackage{algorithm}
\usepackage{algorithmic}

\newtheorem{theorem}{Theorem}

\title{Maximizing Social Welfare in Score-Based Social Distance Games}

\author{%
	Robert Ganian$^1$
\and
	Thekla Hamm$^2$
\and
	Dušan Knop$^3$
\and\\
	Sanjukta Roy$^4$
\and
	Šimon Schierreich$^3$
\and
	Ondřej Suchý$^3$
	\affiliations
	$^1$TU Wien, Austria\\
	$^2$Eindhoven University of Technology, the Netherlands\\
	$^3$Czech Technical University in Prague, Czech Republic\\
	$^4$Penn State University, USA\\
        \emails	
        \{rganian,thekla.hamm\}@gmail.com, dusan.knop@fit.cvut.cz, sanjukta@psu.edu, \{schiesim,ondrej.suchy\}@fit.cvut.cz
}

\usepackage{amsmath}
\usepackage{amssymb}
\usepackage{amsthm}
\usepackage{xspace}
\usepackage{bm}
\usepackage[textwidth=16mm,textsize=footnotesize,disable]{todonotes}
\usepackage{tikz}
\usetikzlibrary{hobby,backgrounds,positioning,calc,trees}
\usepackage{boxedminipage}
\usepackage{cleveref}
\usepackage{enumitem}

\newcommand{\N}{\ensuremath{\mathbb{N}}}

\newcommand{\Z}{\ensuremath{\mathbb{Z}}}

\newcommand{\cdiam}{\ensuremath{{\delta}}}

\newcommand{\bigoh}{\ensuremath{\mathcal{O}}}
\newcommand{\Oh}[1]{\ensuremath{{\bigoh(#1)}}}

\newcommand{\cc}[1]{{\mbox{\textnormal{\textsf{#1}}}}\xspace}    

\newcommand{\FPT}{\cc{FPT}}
\newcommand{\XP}{\cc{XP}}
\newcommand{\NP}{{\cc{NP}}}

\newcommand{\NPh}{\NP-hard\xspace}

\newcommand{\YES}{\cc{Yes}}
\newcommand{\NO}{\cc{No}}
\newcommand{\Yes}{\YES}
\newcommand{\No}{\NO}
\newcommand{\YesI}{\YES-instance\xspace}

\crefname{theorem}{Theorem}{Theorems}

\crefname{observation}{Observation}{Observations}
\newtheorem{lemma}[theorem]{Lemma}
\crefname{lemma}{Lemma}{Lemmas}
\newtheorem{corollary}[theorem]{Corollary}
\crefname{corollary}{Corollary}{Corollaries}
\newtheorem{proposition}[theorem]{Proposition}
\crefname{proposition}{Proposition}{Propositions}

\crefname{conjecture}{Conjecture}{Conjectures}
\newtheorem{claim}{Claim}
\crefname{claim}{Claim}{Claims}
\newenvironment{claimproof}[1]{\par\noindent\underline{Proof:}\space#1}{\hfill $\blacksquare$\smallskip}
\newtheorem{notation}{Notation}

\theoremstyle{remark}
\crefname{example}{Example}{Examples}

\newcommand{\DGs}{\textsc{SDG}\xspace}
\newcommand{\DGNS}{\val-\textsc{SDG-Nash}\xspace}
\newcommand{\DGIR}{\val-\textsc{SDG-IR}\xspace}
\newcommand{\DGWF}{\val-\textsc{SDG-Welfare}\xspace}

\newcommand{\dist}{\ensuremath{\operatorname{dist}}} 
\newcommand{\val}{\ensuremath{\operatorname{s}}\xspace} \newcommand{\maxval}{\ensuremath{v}\xspace}\newcommand{\util}{\ensuremath{\operatorname{u}}\xspace} \newcommand{\SW}{\ensuremath{\operatorname{SW}}\xspace}     \newcommand{\topo}{\ensuremath{\mathsf{T}}\xspace} 
\newcommand{\tw}{\ensuremath{\operatorname{tw}}\xspace} \newcommand{\sz}{\ensuremath{\operatorname{sz}}\xspace} \newcommand{\vc}{\ensuremath{\operatorname{vc}}\xspace}  
\usepackage{etoolbox} 

\begin{document}

\maketitle

\begin{abstract}
Social distance games have been extensively studied as a coalition formation model where the utilities of agents in each coalition were captured using a utility function $u$ that took into account distances in a given social network. In this paper, we consider a non-normalized score-based definition of social distance games where the utility function $u^s$ depends on a generic scoring vector $s$, which may be customized to match the specifics of each individual application scenario. 

As our main technical contribution, we establish the tractability of computing a welfare-maximizing partitioning of the agents into coalitions on tree-like networks, for every score-based function $u^s$. We provide more efficient algorithms when dealing with specific choices of $u^s$ or simpler networks, and also extend all of these results to computing coalitions that are Nash stable or individually rational.
We view these results as a further strong indication of the usefulness of the proposed score-based utility function: even on very simple networks, the problem of computing a welfare-maximizing partitioning into coalitions remains open for the originally considered canonical function $u$.
\end{abstract}

\section{Introduction}

Coalition formation is a central research direction within the fields of algorithmic game theory and computational social choice. While there are many different scenarios where agents aggregate into coalitions, a pervasive property of such coalitions is that the participating agents exhibit \emph{homophily}, meaning that they prefer to be in coalitions with other agents which are close to them. It was this observation that motivated Br{\^{a}}nzei and Larson to introduce the notion of \emph{social distance games} (SDG) as a basic model capturing the homophilic behavior of agents in a social network~\cite{BranzeiL2011}.

Br{\^{a}}nzei and Larson's SDG model consisted of a graph $G=(V,E)$ representing the social network, with $V$ being the agents and $E$ representing direct relationships or connections between the agents. To capture the utility of an agent $v$ in a coalition $C\subseteq V$, the model considered a single function: $u(v,C)=\frac{1}{|C|}\cdot \sum\limits_{w\in C\setminus \{v\}}\frac{1}{d_C(v,w)}$ where $d_C(v,w)$ is the distance between $v$ and $w$ inside $C$.

Social distance games with the aforementioned utility function $u$ have been the focus of extensive study to date, with a number of research papers specifically targeting algorithmic and complexity-theoretic aspects of forming coalitions with maximum social welfare~\cite{BalliuFMO17,KaklamanisKP18,BalliuFMO19,BalliuFMO22}. 
Very recently, Flammini et al.~\cite{FlamminiKOV2020,FlamminiKOV2021} considered a generalization of $u$ via an adaptive real-valued scoring vector which weights the contributions to an agent's utility according to the distances of other agents in the coalition, and studied the price of anarchy and stability for non-negative scoring vectors. 
However, research to date has not revealed any polynomially tractable fragments for the problem of computing coalition structures with maximum social welfare (with or without stability-based restrictions on the behavior of individual agents).

\smallskip
\noindent \textbf{Contribution.}\quad
The undisputable appeal of having an adaptive scoring vector---as opposed to using a single canonical utility function $u$---lies in the fact that it allows us to capture many different scenarios with different dynamics of coalition formation. However, it would also be useful for such a model to be able to assign negative scores to agents at certain (larger) distances in a coalition.
For instance, guests at a gala event may be keen to accept the presence of friends-of-friends (i.e., agents at distance $2$) at a table, while friends-of-friends may be less welcome in private user groups on social networks, and the presence of complete strangers in some scenarios may even be socially unacceptable.

Here, we propose the study of social distance games with a family of highly generic non-normalized score-based utility functions.
Our aim here is twofold. First of all, these should allow us to better capture situations where agents at larger distances are unwelcome or even unacceptable for other agents. 
At the same time, we also want to obtain algorithms capable of computing welfare-maximizing coalition structures in such general settings, at least on well-structured networks.

Our model considers a graph $G$ accompanied with an integer-valued, fixed but adaptive \emph{scoring vector} $\val$ which captures how accepting agents are towards other agents based on their pairwise distance\footnote{Formal definitions are provided in the Preliminaries.}. The utility function $u^{\val}(v,C)$ for an agent $v$ in coalition $C$ is then simply defined as $u^{\val}(v,C)=\sum\limits_{w\in C\setminus \{v\}} \val(d_C(v,w))$; we explicitly remark that, unlike previous models, this is not normalized with respect to the coalition size. As one possible example, a scoring vector of $(1,0,-1)$ could be used in scenarios where agents are welcoming towards friends, indifferent to friends-of-friends, slightly unhappy about friends-of-friends-of-friends (i.e., agents at distance $3$), and unwilling to group up with agents who are at distance greater than $3$ in $G$.
A concrete example which also illustrates the differences to previous SDG models is provided in \Cref{fig:example}.

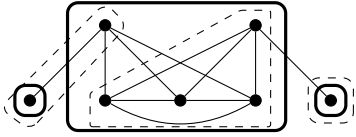
\begin{figure}
	\begin{center}
		\begin{tikzpicture}
			[
			every node/.style={draw, fill=black, shape=circle, inner sep=0pt, text width=1.5mm, align=center, label distance=1mm}
			]
			\node (x1) at  (-2,0) {};
			\node (a1) at  (-1,0) {};
			\node (a2) at  (0,0) {};
			\node (a3) at  (1,0) {};
			\node (y1) at  (2,0) {};
			
			\node (x) at (-1,1) {};
			\node (y) at (1,1) {};
			
			\foreach \i in {a1,a2,a3}
			{
				\draw (x) to (\i) to (y);
			}     
			\draw (x) to (x1);
			\draw (y) to (y1);
			\draw (a1) to (a2) to (a3) to[bend left] (a1);
			
			\draw[very thick,rounded corners] (-1.5,-0.4) rectangle (1.4,1.3); 
			\draw[very thick,rounded corners] (-2.2,-0.2) rectangle (-1.8,0.2); 
			\draw[very thick,rounded corners] (2.2,-0.2) rectangle (1.8,0.2); 
			
			\draw[dashed,rounded corners] (-2,-0.4) -- (-0.7,1) -- (-1.,1.3) -- (-2.4,-0.1) -- cycle;
			\draw[dashed,rounded corners=0.3mm] (-1.2,-0.35) -- (1.2,-0.35) -- (1.2,1.2) -- (0.8,1.2) -- (-1.2,0.1) -- cycle;
			\draw[dashed,rounded corners] (2.3,-0.3) rectangle (1.7,0.3); 
		\end{tikzpicture}
		\caption{A social network illustrating the difference of maximising social welfare in our model compared to previous SDG models. (1) In Br{\^{a}}nzei and Larson's SDG model, the welfare-maximum outcome is the grand coalition. (2) A welfare-maximum outcome in the normalized model of Flammini et al.\ with a scoring vector of $(1,0,0,0)$ is marked with dashed lines, while the same scoring vector in our non-normalized model produces the grand coalition. (3) A scoring vector of \(\val = (1,0,-1)\) in our model produces the welfare-maximizing outcome marked with bold lines, with a welfare of $18$. (4) A `less welcoming' scoring vector of \(\val = (1,-3)\) leads to the welfare maximizing dash-circled partition with a welfare of \(14\) (compared to only \(12\) for the bold-circled one).
			\label{fig:example}}
	\end{center}
\end{figure}

While non-normalized scoring functions have not previously been considered for social distance games, we view them a natural way of modeling agent utilities; in fact, similar ideas have been successfully used in models for a variety of other phenomena including, e.g., committee voting~\cite{ElkindI15}, resource allocation~\cite{BouveretL08,BouveretCM16} 
and Bayesian network structure learning~\cite{OrdyniakS13,GanianK21}. 
Crucially, it is not difficult to observe that many of the properties originally established by Br{\^{a}}nzei and Larson for SDGs also hold for our non-normalized score-based model with every choice of $s$, such as the small-world property~\cite{jackson2008social,BranzeiL2011} and the property that adding an agent with a close (distant) connection to a coalition positively (negatively) impacts the utilities of agents~\cite{BranzeiL2011}.
In addition, the proposed model can also directly capture the notion of \emph{enemy aversion} with symmetric preferences~\cite{OhtaBISY17,BarrotY19} by setting $s=(1)$.

Aside from the above, a notable benefit of the proposed model lies on the complexity-theoretic side of things. Indeed, a natural question that arises in the context of SDG is whether we can compute an outcome---a partitioning of the agents into coalitions---which maximizes the social welfare (defined as the sum of the utilities of all agents in the network). This question has been studied in several contexts, and depending on the setting one may also require the resulting coalitions to be stable under \emph{individual rationality} (meaning that agents will not remain in coalitions if they have negative utility) or \emph{Nash stability} (meaning that agents may leave to join a different coalition if it would improve their utility). But in spite of the significant advances in algorithmic aspects of other coalition formation problems in recent years~\cite{BoehmerE20,BoehmerE20b,ChenGH20,GanianHKSS22}, 
we lack any efficient algorithm capable of producing such a welfare-optimal partitioning when using the utility function $u$ even for the simplest types of networks.

To be more precise, when viewed through the refined lens of \emph{parameterized complexity}~\cite{DowneyF13,CyganFKLMPPS2015} that has recently become a go-to paradigm for such complexity-theoretic analysis, no tractable fragments of the problem are known. More precisely, the problem of computing a welfare-maximizing outcome under any of the previously considered models is not even known to admit an \XP\ algorithm when parameterized by the minimum size of a vertex cover in the social network $G$---implying a significant gap towards potential fixed-parameter tractability. This means that the complexity of welfare-maximization under previous models remains wide open even under the strongest non-trivializing restriction of the network.

As our main technical contribution, we show that non-normalized score-based utility functions do not suffer from this drawback and can in fact be computed efficiently under fairly mild restrictions on $G$. Indeed, as our first algorithmic result we obtain an \XP\ algorithm that computes a welfare-maximizing partitioning of the agents into coalitions parameterized by the treewidth of $G$, and we strengthen this algorithm to also handle additional restrictions on the coalitions in terms of individual rationality or Nash stability.
We remark that considering networks of small treewidth is motivated not only by the fundamental nature of this structural graph measure, but also by the fact that many real-world networks exhibit bounded treewidth~\cite{ManiuSJ19}.

In the next part of our investigation, we show that when dealing with simple scoring functions or bounded-degree networks, these results can be improved to fixed-parameter algorithms for welfare-maximization (including the cases where we require the coalitions to be individually rational or Nash stable). This is achieved by combining structural insights into the behavior of such coalitions with a different dynamic programming approach. Furthermore, we also use an entirely different technique based on quadratic programming to establish the fixed-parameter tractability of all 3 problems under consideration w.r.t.\ the minimum size of a vertex cover in $G$. Finally, we conclude with some interesting generalizations and special cases of our model and provide some preliminary results in these directions. 

\ifshort
\smallskip
\noindent{\itshape Statements where proofs or more details are provided in the
appendix are marked with $\star$.}
\fi

\section{Preliminaries}
We use $\N$ to denote the set of natural numbers, i.e., positive integers, and $\Z$ for the set of integers. For $i\in \N$, we let $[i]= \{1,\ldots,i\}$ and ${[i]_0 = [i] \cup \{0\}}$. 
\iflong
We assume basic familiarity with graph-theoretic terminology~\cite{Diestel17}.
\fi
\ifshort
We assume basic familiarity with graph-theoretic terminology~\cite{Diestel17} and the parameterized complexity paradigm~\cite{DowneyF13,CyganFKLMPPS2015}, notably the complexity classes \FPT\ and \XP.
\fi

\paragraph{Social Distance Games.}
A \emph{social distance game} (SDG) consists of a set $N = \{1,\ldots,n\}$ of \emph{agents}, a simple undirected graph $G=(N,E)$ over the set of agents called a \emph{social network}, and a non-increasing \emph{scoring vector} $s=(s_1,\dots,s_{\cdiam})$ where
\begin{itemize}
\item for each $a\in [\cdiam]$, $s_a\in \Z$ and 
\item for each $a\in [\cdiam-1]$, $s_{a+1}\leq s_a$.
\end{itemize}

In some cases, it will be useful to treat $s$ as a function from $\N$ rather than a vector; to this end, we set $s(a)=s_a$ for each $a\leq \cdiam$ and $s(a)=-\infty$ when $a>\cdiam$. 
The value ``$-\infty$'' here represents an inadmissible outcome, and formally we set $-\infty+z=-\infty$ and $-\infty<z$ for each $z\in \Z$. 
Furthermore, let $\maxval=\maxval(\val)=\val(1)=\max_{r=1}^{\infty} \val(r)$.

A \emph{coalition} is a subset $C\subseteq N$, and an outcome is a partitioning $\Pi=(C_1,\dots,C_\ell)$ of $N$ into coalitions; formally,
$\bigcup_{i=1}^\ell C_i = N$, every $C_i\in \Pi$ is a coalition, and all coalitions in $\Pi$ are pairwise disjoint. We use $\Pi_i$ to denote the coalition the agent $i\in N$ is part of in the outcome $\Pi$.
The \emph{utility} of an agent $i\in N$ for an outcome $\Pi$ is
\[
\util^{\val}(i,\Pi) = \sum_{j\in \Pi_i\setminus\{i\}} \val(\dist_{\Pi_i}(i,j)),
\]
where $\dist_{\Pi_i}(i,j)$ is the length of a shortest path between $i$ and $j$ in the graph $G[\Pi_i]$, i.e., the subgraph of $G$ induced on the agents of $\Pi_i$. We explicitly note that if $\Pi_i$ is a singleton coalition then $\util^{\val}(i,\Pi_i)=0$. Moreover, in line with previous work~\cite{BranzeiL2011} we set $\dist_{\Pi_i}(i,j):=-\infty$ if there is no $i$-$j$ path in $G[\Pi_i]$, meaning that $\util^{\val}(i,\Pi_i)=-\infty$ whenever $G[\Pi_i]$ is not connected.

For brevity, we drop the superscript from $u^{\val}$ whenever the scoring vector \val\ is clear from the context. To measure the satisfaction of the agents with a given outcome, we use the well-known notation of \emph{social welfare}, which is the total utility of all agents for an outcome $\Pi$, that is,
\[
\SW^{\val}(\Pi) = \sum_{i\in N} \util^{\val}(i,\Pi).
\]
Here, too, we drop the superscript specifying the scoring vector whenever it is clear from the context.

We assume that all our agents are selfish, behave strategically, and their aim is to maximize their utility. To do so, they can perform \emph{deviations} from the current outcome $\Pi$. We say that $\Pi$ admits an \emph{IR-deviation} if there is an agent $i\in N$ such that $\util(i,C) < 0$; in other words, agent $i$ prefers to be in a singleton coalition over its current coalition. If no agent admits an IR-deviation, the outcome is called \emph{individually rational} (IR). We say that $\Pi$ admits an \emph{NS-deviation} if there is an agent $i$ and a coalition $C\in \Pi\cup \{\emptyset\}$ such that $\util(i,C\cup\{i\}) > \util(i,\Pi_i)$. $\Pi$~is called \emph{Nash stable} (NS) if no agent admits an NS-deviation. 
We remark that other notions of stability exist in the literature~\cite[Chapter 15]{BrandtCELP16}, but Nash stability and individual rationality are the most basic notions used for stability based on individual choice~\cite{SungD07,Karakaya11}.

Having described all the components in our score-based SDG model, we are now ready to formalize the three classes of problems considered in this paper. We note that even though these are stated as decision problems for complexity-theoretic reasons, each of our algorithms for these problems can also output a suitable outcome as a witness. For an arbitrary fixed scoring vector $\val$, we define:

\begin{center}
	\begin{boxedminipage}{0.98 \columnwidth}
		\DGWF\\[5pt]
		\begin{tabular}{l p{0.78 \columnwidth}}
			Input: & A social network $G=(N,E)$, desired welfare $b \in \N$.\\
			Question: \hspace{-0.4cm} & Does the distance game given by $G$ and $\val$ admit an outcome with social welfare at least~$b$?
		\end{tabular}
	\end{boxedminipage}
\end{center}
\iflong
\begin{center}
	\begin{boxedminipage}{0.98 \columnwidth}
		\DGIR\\[5pt]
		\begin{tabular}{l p{0.78 \columnwidth}}
			Input: & A social network $G=(N,E)$, desired welfare $b \in \N$.\\
			Question: \hspace{-0.4cm} & Does the distance game given by $G$ and $\val$ admit an individually rational outcome with social welfare at least $b$?
		\end{tabular}
	\end{boxedminipage}
\end{center}

\begin{center}
	\begin{boxedminipage}{0.98 \columnwidth}
		\DGNS\\[5pt]
		\begin{tabular}{l p{0.78 \columnwidth}}
			Input: & A social network $G=(N,E)$, desired welfare $b \in \N$.\\
			Question: \hspace{-0.4cm} & Does the distance game given by $G$ and $\val$ admit a Nash stable outcome with social welfare at least $b$?
		\end{tabular}
	\end{boxedminipage}
\end{center}
\fi
\ifshort
\DGIR\ and \DGNS\ are then defined analogously, but with the additional condition that the outcome must be individually rational or Nash stable, respectively.
\fi

We remark that for each of the three problems, one may assume w.l.o.g.\ that $s(1)>0$; otherwise the trivial outcome consisting of $|N|$ singleton coalitions is both welfare-optimal and stable.
Moreover, w.l.o.g.\ we assume $G$ to be connected since an optimal outcome for a disconnected graph $G$ can be obtained as a union of optimal outcomes in each connected component of $G$.

The last remark we provide to the definition of our model is that it trivially also supports the 
well-known \emph{small world} property~\cite{jackson2008social} that has been extensively studied on social networks. 
In their original work on SDGs, Br{\^{a}}nzei and Larson showed that their model exhibits the small world property by establishing a diameter bound of $14$ in each coalition in a so-called \emph{core partition}~\cite{BranzeiL2011}.
Here, we observe that for each choice of $\val$, a welfare-maximizing coalition will always have diameter at most $\cdiam$.

\iflong
\paragraph{Parameterized Complexity.} The \emph{parameterized complexity} framework~\cite{DowneyF13,CyganFKLMPPS2015} provides the ideal tools for the fine-grained analysis of computational problems which are \NPh\ and hence intractable from the perspective of classical complexity theory. Within this framework, we analyze the running times of algorithms not only with respect to the input size $n$, but also with respect to a numerical parameter $k\in\N$ that describes a well-defined structural property of the instance; the central question is then whether the superpolynomial component of the running time can be confined by a function of this parameter alone. 

The most favorable complexity class in this respect is \FPT (short for ``fixed-parameter tractable'') and contains all problems solvable in $f(k)\cdot n^\Oh{1}$ time, where $f$ is a computable function. Algorithms with this running time are called \emph{fixed-parameter algorithms}. A less favorable, but still positive, outcome is an algorithm with running time of the form~$n^f(k)$; problems admitting algorithms with such running times belong to the class \XP.
\fi

\paragraph{Structural Parameters.} Let $G=(V,E)$ be a graph. A set $U\subseteq V$ is a \emph{vertex cover} if for every edge $e\in E$ it holds that $U\cap e \not= \emptyset$. The \emph{vertex cover number} of $G$, denoted $\vc(G)$, is the minimum size of a vertex cover of $G$.

A \emph{nice tree-decomposition} of $G$ is a pair $(\mathcal{T},\beta)$, where $\mathcal{T}$ is a tree rooted at a node $r\in V(\mathcal{T})$, $\beta\colon V(\mathcal{T})\to 2^{V}$ is a function assigning each node $x$ of $\mathcal{T}$ its \emph{bag}, and the following conditions hold:
\begin{itemize}
	\item for every edge $\{u,v\}\in E(G)$ there is a node $x\in V(\mathcal{T})$ such that $u,v\in\beta(x)$,
	\item for every vertex $v\in V$, the set of nodes $x$ with $v\in\beta(x)$ induces a connected subtree of $\mathcal{T}$,
	\item $|\beta(r)|=|\beta(x)| = 0$ for every \emph{leaf} $x\in V(\mathcal{T})$, and
	\item there are only tree kinds of internal nodes in $\mathcal{T}$:
		\begin{itemize}
			\item $x$ is an \emph{introduce node} if it has exactly one child $y$ such that $\beta(x) = \beta(y)\cup\{v\}$ for some ${v\notin\beta(y)}$,
			\item $x$ is a \emph{join node} if it has exactly two children $y$ and $z$ such that $\beta(x) = \beta(y) = \beta(z)$, or
			\item $x$ is a \emph{forget node} if it has exactly one child $y$ such that $\beta(x) = \beta(y)\setminus\{v\}$ for some $v\in\beta(y)$.
		\end{itemize}
\end{itemize}
The \emph{width} of a nice tree-decomposition $(\mathcal{T},\beta)$ is $\max_{x\in V(\mathcal{T})} |\beta(x)|-1$, and the treewidth $\tw(G)$ of a graph $G$ is the minimum width of a nice tree-decomposition of $G$. Given a nice tree-decomposition and a node~$x$, we denote by $G^x$ the subgraph induced by the set $V^x = \bigcup_{y\text{ is a descendant of }x}\beta(y)$, where we suppose that~$x$ is a descendant of itself. 
\ifshort
It is well-known that optimal nice tree-decompositions can be computed efficiently~\cite{Kloks94,Bodlaender96,Korhonen21}. ($\star$)
\fi
\iflong
It is well-known that computing a nice tree-decomposition of minimum width is fixed-parameter tractable when parameterized by the treewidth~\cite{Kloks94,Bodlaender96}, and even more efficient algorithms exist for obtaining near-optimal nice tree-decompositions~\cite{Korhonen21}.
\fi

\paragraph{Integer Quadratic Programming.}
\textsc{Integer Quadratic Programming} (IQP) over $d$ dimensions can be formalized as the task of computing
\begin{equation}\label{eq:generalIQP}\tag{IQP}
	\max \left\{ x^{T} Q x \mid A x \le b,\, x \ge 0 ,\, x \in \mathbb{Z}^d \right\} \,,	
\end{equation}
where $Q \in \mathbb{Z}^{d \times d}$, $A  \in \mathbb{Z}^{m \times d}$, $b \in \mathbb{Z}^{m}$.
That is, IQP asks for an integral vector $x \in \mathbb{Z}^d$ which maximizes the value of a quadratic form subject to satisfying a set of linear constraints.

\begin{proposition}[{\cite{Lokshtanov15,Zemmer2017}}, see also~\cite{GavenciakKK22}]\label{prop:IQPisFPT}
\textsc{Integer Quadratic Programming} is fixed-parameter tractable when parameterized by $d+\|A\|_{\infty}+ \|Q\|_{\infty}$. 
\end{proposition}

\section{Structural Properties of Outcomes}
\label{sec:structure}

As our first set of contributions, we establish some basic properties of our model and the associated problems that are studied within this paper. We begin by showcasing that the imposition of individual rationality or Nash stability as additional constraints on our outcomes does in fact have an impact on the maximum welfare that can be achieved (and hence it is indeed necessary to consider three distinct problems). We do not consider this to be obvious at first glance: intuitively, an agent $i$'s own contribution to the social welfare can only improve if they perform an IR- or NS-deviation, and the fact that the distance function $\dist_{\Pi_i}$ is symmetric would seem to suggest that this can only increase the total social welfare.

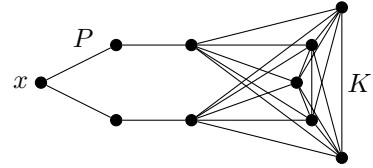
\begin{figure}
 \begin{center}
  \begin{tikzpicture}
      [
      every node/.style={draw, fill=black, shape=circle, inner sep=0pt, text width=1.5mm, align=center, label distance=1mm}
      ]
      \node[label=left:$x$] (x) at  (0,0) {};
      \node (p2)  at (1,0.5) {};
      \node (p4) at (1,-0.5) {};
      \node (p1) at (2,0.5) {};
      \node (p5) at(2,-0.5) {};
      \node (1) at  (4,1) {};
      \node (5) at (4,-1) {};
      \node (2) at (3.6,0.5) {};
      \node (4) at (3.6,-0.5) {};
      \node (3) at (3.4,0) {};
      \draw (p1) to (p2) to node[above=2mm,draw=none,fill=none] {$P$} (x) to (p4) to (p5);
      \draw (1) to (2) to (3) to (4) to (5) to node[right,draw=none,fill=none] {$K$} (1) to (3) to (5) to (2) to (4) to (1);
      \foreach \i in {1,...,5}
      {
        \draw (p1) to (\i) to (p5);
      }      
  \end{tikzpicture}
  \caption{Social Network from \Cref{lem:irrational}.}\label{fig:irrational}
 \end{center}
\end{figure}

\begin{figure}
 \begin{center}
  \begin{tikzpicture}
      [
      every node/.style={draw, fill=black, shape=circle, inner sep=0pt, text width=1.5mm, align=center, label distance=1mm}
      ]
      \node[label=above:$x$] (x) at  (0,0) {};
      \node (p2)  at (1,0.5) {};
      \node (p4) at (1,-0.5) {};
      \node (p1) at (2,0.5) {};
      \node (p5) at(2,-0.5) {};
      \node (1) at  (4,0.7) {};
      \node (5) at (4,-0.7) {};
      \node (2) at (3.6,0.3) {};
      \node (4) at (3.6,-0.3) {};
      \node[label=above:$y$] (y) at (-1,0) {};
      \draw (p1) to (p2) to node[above=2mm,draw=none,fill=none] {$P$} (x) to (p4) to (p5);
      \draw (x) to (y);
      \draw (1) to (2) to (4) to (5) to node[right,draw=none,fill=none] {$K$} (1) to (4);
      \draw (2) to (5);
      \foreach \i in {1,2,4,5}
      {
        \draw (p1) to (\i) to (p5);
      }      
  \end{tikzpicture}
  \caption{Social Network from \Cref{lem:non-nash}.}\label{fig:non-nash}
 \end{center}
\end{figure}
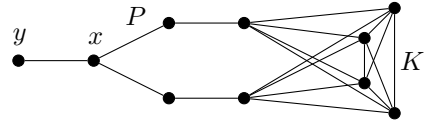

\iflong
\begin{lemma}
\fi
\ifshort
\begin{lemma}[$\star$]
\fi
\label{lem:irrational}
 There is a scoring vector $\val$ and a social network $G$ such that the single outcome achieving the maximum social welfare is not individually rational.
\end{lemma}

\iflong
\begin{proof}
 Consider a scoring function $\val$ such that $\val=(1,1,-1,-1,-1,-1)$.
  Consider the social network $G$ in \Cref{fig:irrational} formed from a path $P$ on $5$ vertices and a clique $K$ on $5$ vertices by connecting the endpoints of~$P$ to all vertices of $K$.
 Let $x$ be the central agent of $P$.
 Let~$C$ be the grand coalition in $G$.
 The graph can be viewed as a $6$-cycle with $K$ forming one ``bold'' agent.
 All vertices on the cycle contribute positively to the agent's utility, except for the one that is exactly opposite on the cycle.
 Hence, $\util(x,C)=4-5=-1$, while utility of all other agents is $8-1=7$ in $C$.
 This gives total social welfare of $62$ for the grand coalition.

 However, if $x$ leaves the coalition to form its own one, their utility will improve from $-1$ to $0$, whereas the total social welfare drops. 
 Indeed, in $C \setminus \{x\}$ there are 2 agents with utility $6-2=4$, 2 agents with utility $7-1=6$ and 5 agents with utility $8-0$, giving total social welfare of $60$.
 If any $y\neq x$ was to be excluded from $C$ to form outcome $\{y\}, C\setminus \{y\}$, then $y$ joining $C$ improves social welfare, proving that it was not optimal.
 Finally, if the outcome consists of several coalitions with the largest one of size 8, then the welfare is at most $8 \cdot 7+2 \cdot 1= 56$, if the largest size is 7, then we get at most $7 \cdot 6+3\cdot 2=48$, for 6 it is $6\cdot 5+4\cdot 3=42$ and for 5 it is $5 \cdot 4 +5 \cdot 4=40$.

 Hence the grand coalition $C$ is the only outcome with maximal social welfare, but it is not individually rational (and therefore not Nash stable), as $\util(x,C)=-1$.
\end{proof}
\fi

\iflong
\begin{lemma}
\fi
\ifshort
\begin{lemma}[$\star$]
\fi
\label{lem:non-nash}
There is a scoring vector $\val$ and a social network $G$ such that the single individually rational outcome achieving the maximum social welfare among such outcomes is not Nash stable.
\end{lemma}

\iflong
\begin{proof}
 Consider again the scoring function $\val=(1,1,-1,-1,-1,-1)$.
  Similarly to previous lemma, consider the social network $G$ in \Cref{fig:non-nash} formed from a path $P$ on $5$ vertices and a clique $K$ on $4$ vertices by connecting the endpoints of~$P$ to all vertices of $K$ and adding a agent $y$ only connected to the central agent of $P$ which we call $x$.
 Let~$C$ be the coalition containing all vertices of $G$ except for $y$.
 As in the previous lemma, $G[C]$ can be viewed as a $6$-cycle with $K$ forming one ``bold'' agent.
 Hence, $\util_x(C)=4-4=0$, while utility of other agents in $C$ is $7-1=6$.
 Trivially $\util_y(\{y\})=0$, hence the outcome $(\{y\},C)$ is individually rational.
 It has total social welfare of $48$.
 However, it is not Nash stable, as $x$ wants to deviate to $\{x,y\}$ giving them utility $1$.

 However, the outcome $(\{x,y\}, C\setminus\{x\})$, which is Nash stable, has total social welfare only $46$.
 Note that $\util_z(C\setminus\{x\}) \ge 3$ for every agent $z \in C\setminus\{x\}$, so any outcome $(\{x,y,z\}, C\setminus\{x,z\})$ cannot be Nash stable.
 While the total social welfare of the grand coalition is $46$, the utility of $y$ is $3-6=-3$ in this coalition, so this outcome is not even individually rational.
 From the computations in the previous lemma, it follows, that to attain the social welfare of $48$, the largest coalition in the outcome must be of size at least \(7\). 
 Moreover, if it is of size exactly \(7\), then these \(7\) vertices must be at mutual distance at most $2$. 
 However, there are no \(7\) vertices in mutual distance at most \(2\) in $G$.
 Hence, in any outcome with social welfare $48$ the largest coalition must be of size at least $8$.
 Agent $y$ has only \(3\) agents in distance at most \(2\) in $G$.
 Hence, for $y$ to get a positive utility from some coalition, the coalition must be of size at most \(7\), i.e., $y$ cannot be part of the largest coalition in any outcome with social welfare at least \(48\).
 However, for every $z \in C$, $z$ joining the coalition $C\setminus \{z\}$ improves the social welfare of the outcome, proving that it was not optimal.
 
 Hence the outcome $(\{y\},C)$ is the only individually rational outcome with maximal social welfare, but it is not Nash stable.
\end{proof}
\fi

\ifshort
Both Lemma~\ref{lem:irrational} and Lemma~\ref{lem:non-nash} can in fact be shown to hold for the same scoring vector $\val=(1,1,-1,-1,-1,-1)$; for the former we use the network depicted in Figure~\ref{fig:irrational} while for the latter the network in Figure~\ref{fig:non-nash}.
\fi

As our next two structural results, we prove that on certain SDGs it is possible to bound not only the diameter but also the size of each coalition in a welfare-maximum outcome. Notably, we establish such bounds for SDGs on bounded-degree networks and SDGs which have a simple scoring vector on a tree-like network. While arguably interesting in their own right, these properties will be important for establishing the fixed-parameter tractability of comptuing welfare-optimal outcomes in the next section.

\begin{lemma}\label{lem:maxdeg_coal_size}
 If $G$ is a graph of maximum degree $\Delta(G)$ and $C$ is a coalition of size more than $(v+1) \cdot \Delta(G) \cdot (\Delta(G)-1)^{\cdiam-2}$, then for every $i \in C$ we have $\util(i,C) <0$.
\end{lemma}

\begin{proof}
 Let $i \in C$.
 There are at most $\Delta(G) \cdot (\Delta(G)-1)^{\cdiam-2}$ agents in distance at most $\cdiam-1$ from $i$.
 Each of these agents contributes at most $\maxval$ to $\util(i,C)$.
 Every other agent contributes at most $-1$.
 Hence, if there are more than $(\maxval+1) \cdot \Delta(G) \cdot (\Delta(G)-1)^{\cdiam-2}$ agents in $C$, then
 more than $\maxval \cdot \Delta(G) \cdot (\Delta(G)-1)^{\cdiam-2}$ of them have a negative contribution to $\util(i,C)$ and
 \begin{multline*}
 \util(i,C) < \maxval \cdot \Delta(G) \cdot (\Delta(G)-1)^{\cdiam-2} \\
 -1 \cdot  \maxval \cdot \Delta(G) \cdot (\Delta(G)-1)^{\cdiam-2} =0. \qedhere  
 \end{multline*}

\end{proof}

\begin{lemma}\label{lem:degen_coal_size}
Let $\val$ be such that $\val(2) < 0$.
 If $G$ is a graph of treewidth $\tw$ and $C$ is a coalition of size more than $2(\maxval+1) \cdot \tw + 1$, then $\sum_{i \in C}\util(i,C) <0$.
\end{lemma}

\begin{proof}
 Each agent adjacent to $i$ contributes $\maxval$ to $\util(i,C)$, whereas all the other agents contribute at most $-1$.
 Since a  graph of treewidth $\tw$ is $\tw$-degenerate, there are $|E(G[C])| \le |C| \cdot \tw$ pairs of adjacent agents and $\binom{|C|}{2} - |E(G[C])|$ pairs of non-adjacent agents.
 We have 
 \begin{align*}
  \sum_{i \in C}\util(i,C) 
      &= \sum_{i,j \in C; i\neq j}\val\left(\dist(i,j)\right)\\
      &\le 2\left(\maxval\cdot \left|E\left(G[C]\right)\right| - \left(\binom{|C|}{2} - \left|E\left(G[C]\right)\right|\right)\right)\\
      &= 2\left((\maxval+1) \cdot  \left|E\left(G[C]\right)\right| - \binom{|C|}{2}\right)\\
      &\le 2(\maxval+1) \cdot |C| \cdot \tw - |C|(|C|-1)\\
      &=|C|\left(2(\maxval+1) \cdot \tw- (|C|-1)\right)\\
      &<|C|\left(2(\maxval+1) \cdot \tw -\left(2(\maxval+1) \cdot \tw+ 1-1\right)\right)\\
      &=0. \qedhere
 \end{align*}
\end{proof}

\section{Computing Optimal Outcomes}
\label{sec:algo}

\subsection{Intractability}
As our first step towards an understanding of the complexity of computing a welfare-optimal outcome in a SDG, we establish the \NP-hardness of \DGWF, \DGIR\ and \DGNS\ even for a very simple choice of $s$.

\ifshort
\begin{theorem}[$\star$]
\label{thm:NPh}
Let $\val=(\maxval)$ for any $\maxval>0$.
Then \DGWF, \DGIR\ and \DGNS are \NPh.
\end{theorem}

\begin{proof}[Proof Sketch]
As our first step, we prove the \NP-hardness of the following intermediate problem via an adaptation of a known reduction from \textsc{NotAllEqual-3-SAT}~\cite[Theorem 9.8]{Papadimi94}:
\begin{center}
\begin{boxedminipage}{0.98 \columnwidth}
\textsc{3-Coloring Triangle Covered Graph  (3CTCG)}\\[2pt]
\begin{tabular}{l p{0.78 \columnwidth}}
Input: & An undirected graph $G=(V,E)$ with $|V|=3n$ vertices such that $G$ contains a collection of $n$ mutually vertex disjoint triangles.\\
Question: \hspace{-0.4cm} & Does $G$ have a 3-coloring?
\end{tabular}
\end{boxedminipage}
\end{center}

Next, we reduce \textsc{3CTCG} to our three problems via a single construction. Let $G$ be an instance of \textsc{3CTCG} with $3n$ vertices and $T_1, \ldots, T_n$ the corresponding collection of triangles.
Let $\overline{G}$ be a complement of $G$, let $\maxval=\maxval(\val)$ and let $b=3n\maxval\cdot(n-1)$.
To establish the \NP-hardness of \DGWF, it suffices to show that $G$ is a \YesI{} of \textsc{3CTCG} if and only if $\overline{G}$ admits an outcome with social welfare at least $b$; for the remaining two problems, we additionally show that such an outcome will furthermore be individually rational and Nash stable. 
\end{proof}
\fi

\iflong
\begin{theorem}
\label{thm:NPh}
Let \val be such that $\val=(\maxval)$ for any $\maxval>0$.
Then \DGWF, \DGNS and \DGIR are all \NPh.
\end{theorem}
\begin{proof}
We first prove the \NP-hardness of the following problem.
\begin{center}
\begin{boxedminipage}{0.98 \columnwidth}
\textsc{3-Coloring Triangle Covered Graph  (3CTCG)}\\[2pt]
\begin{tabular}{l p{0.78 \columnwidth}}
Input: & An undirected graph $G=(V,E)$ with $|V|=3n$ vertices such that $G$ contains a collection of $n$ mutually vertex disjoint triangles.\\
Question: \hspace{-0.4cm} & Does $G$ have a 3-coloring?
\end{tabular}
\end{boxedminipage}
\end{center}
To show that this problem is \NPh, we use the textbook reduction from \textsc{NAE-3-Sat} to \textsc{3-Coloring} \cite[Theorem 9.8]{Papadimi94} with a slight modification.
Let $\varphi$ be a formula with variables $x_1, \ldots,x_n$ and clauses $C_1, \ldots, C_m$ forming an instance of \textsc{NAE-3-Sat}.
We construct a graph $G$ as follows.
For each variable $x_i$ we introduce a triangle $a_i,x_i,\overline{x_i}$ an we add the edges between $x_i,\overline{x_i}$ and $a_{i+1}$ for every $i \in \{1, \ldots, n-1\}$.
For each clause $C_j$ we add a triangle $C_{j1},C_{j2},C_{j3}$ and connect the vertex $C_{jr}$ with the vertex representing the $r$-th literal of clause $C_j$.

Clearly, as each vertex was introduced together with its triangle, graph $G$ is a valid instance of \textsc{3CTCG}.

If $h$ is a scoring function of $x_1, \ldots,x_n$ that nae-satisfies $\varphi$, then we color $G$ with colors $b,t,f$ as follows.
We color all $a_i$'s with a color $b$. For each $i \in \{1, \ldots, n\}$ if $h(x_i)=\text{true}$, then we color $x_i$ with color $t$ and $\overline{x_i}$ with color $f$ and otherwise we color $x_i$ with color $f$ and $\overline{x_i}$ with color $t$.
For each $j \in \{1, \ldots, m\}$ there is $r_t$ such that the $r_t$-th literal of $C_j$ evaluates to true under $h$ and $r_f$ such that the $r_f$-th literal of $C_j$ evaluates to false under $h$. We color $C_{jr_t}$ by $f$, $C_{jr_f}$ by $t$, and the remaining vertex of triangle $C_{j1},C_{j2},C_{j3}$ by $b$. This way we obtain a 3-coloring of $G$.

Now assume that there is a $3$-coloring $c$ of $G$ by colors $b,f,t$. Assume without loss of generality that $c(a_1)=b$.
Since the vertices $x_1,\overline{x_1}$ use colors $f$ and $t$, we have also $c(a_2)=b$ and $c(a_i)=b$ for every $i \in \{1, \ldots, n\}$. Only colors $t,f$ are used on $x_i$'s and $\overline{x_j}$'s. Consider a scoring function $h$ of $x_1, \ldots,x_n$ such that $h(x_i)$ is true if and only if $c(x_i)=t$.
We claim that $h$ nae-satisfies $\varphi$. Suppose not and $C_j$ is not satisfied by $h$. This means that all the literals of $C_j$ are colored in the same color by $c$, either $t$ or $f$. In either case $c$ would have to color all tree vertices of the triangle $C_{j1},C_{j2},C_{j3}$ by the two remaining colors, which is impossible. Hence $h$ nae-satisfies $\varphi$.

Now we reduce \textsc{3CTCG} to \DGWF. Let $G$ be an instance of \textsc{3CTCG} with $3n$ vertices and $T_1, \ldots, T_n$ the corresponding collection of triangles.
Let $\overline{G}$ be a complement of $G$, let $\maxval=\maxval(\val)$ and let $b=3n\maxval\cdot(n-1)$.
We claim that $G$ is a \YesI{} of \textsc{3CTCG} if and only if $(\overline{G}, b)$ is a \YesI{} of \DGWF.

If $c$ is a 3-coloring of $G$, then as each $T_i$ contains each color exactly once, exactly $n$ vertices are colored with each color. Let $C_1,C_2,C_3$ be the sets of vertices colored with color $1$, $2$, and $3$, respectively.
Since each $C_i$ forms a clique of size~$n$ in~$\overline{G}$, the utility of each agent is $\maxval \cdot (n-1)$ and the total welfare is $3n\maxval\cdot(n-1)$. Note that, as the utility of each agent is positive, the outcome is individually rational. Moreover, the utility of each agent would be $-\infty$ in any other coalition than she is, so the outcome is also Nash stable.

If $\Pi$ is an outcome with total welfare at least $b$, then each $C \in \Pi$ is a clique, as otherwise some agents would have utility $\val(2)=-\infty$. Therefore, $C$ can contain at most one agent from each $T_i$ and, hence, $|C| \le n$ for each $C \in \Pi$. This implies $\util_i(\Pi_i) \le \maxval\cdot(n-1)$ for each $i$. However, as total welfare is (at least) $b = 3n\maxval\cdot(n-1)$, we have that the utility of each agent is exactly $\maxval\cdot(n-1)$ and, thus, the size of each $C \in \Pi$ is exactly $n$. This in turn implies that there are exactly $3$ coalitions in $\Pi$, let us call them $C_1,C_2,C_3$. It remains to color each agent from coalition $C_1$ with color $1$, each agent from coalition $C_2$ with color $2$, and each agent from coalition $C_3$ with color $3$. This is a $3$-coloring since each of the coalitions is a clique.

It is easy to see that both of the reductions can be performed in polynomial time.
\end{proof}
\fi

\subsection{An Algorithm for Tree-Like Networks}
We complement Theorem~\ref{thm:NPh} by establishing that all three problems under consideration can be solved in polynomial time on networks of bounded treewidth---in other words, we show that they are \XP-tractable w.r.t.\ treewidth. As with numerous treewidth-based algorithms, we achieve this result via leaf-to-root dynamic programming along a tree-decomposition. However, the records we keep during the dynamic program are highly non-trivial and require an advanced branching step to correctly pre-computed the distances in the stored records. 

We first describe the ``baseline'' algorithm for solving \DGWF, and then prove that this may be adapted to also solve the other two problems by expanding on its records and procedures.

\iflong
\begin{theorem}
\fi
\ifshort
\begin{theorem}[$\star$]
\fi
\label{thm:tw}
	For every fixed scoring vector~$\val$, the \DGWF problem is in~\XP when parameterized by the treewidth of the social network~$G$.
\end{theorem}
\ifshort
\begin{proof}[Proof Sketch]
		
	\newcommand{\Pttn}{\ensuremath{C}} 	\newcommand{\len}{\ensuremath{\operatorname{len}}}
	\newcommand{\Spaths}{\ensuremath{S}} 	\newcommand{\Dvec}{\ensuremath{T}} 		Our algorithm is based on leaf-to-root dynamic programming along a nice tree-decomposition of the input social network. In each node $x$ of the tree-decomposition, we store a set $\mathcal{R}_x$ of partial solutions called \emph{records}. Each record realizes a single \emph{signature} which is a triple~$(\Pttn,\Spaths,\Dvec)$, where
	\begin{itemize}
		\item $\Pttn$ is a partition of bag agents into parts of coalitions; there are at most $\tw+1$ different coalitions intersecting $\beta(x)$ and, thus, at most ${tw^\Oh{\tw}}$ possible partitions of $\beta(x)$.
		\item $\Spaths$ is a function assigning each pair of agents that are part of the same coalition according to $\Pttn$ the shortest intra-coalitional path; recall that for fixed \val, the diameter of every coalition is bounded by a constant~$\cdiam$ and, therefore, there are~${n^\Oh{\cdiam} = n^\Oh{1}}$ possible paths for each pair of agents which gives us~${n^\Oh{\tw^2}}$ combinations in total.
		\item $\Dvec$ is a table storing for every coalition $P$ and every possible vector of distances to bag agents that are in $P$ the number of agents from $P$ that were already forgotten in some node of the tree-decomposition; the number of possible coalitions is at most $\tw+1$, the number of potential distance vectors is $\cdiam^{\tw+1} = 2^\Oh{\tw}$, and there are at most $n$ values for every combination of coalition and distance vector which leads to at most~${n^\Oh{2^{\tw}}}$ different tables~$\Dvec$.
	\end{itemize}
	The value of every record is a pair $(\pi,w)$, where $\pi$ is a partition of~$V^x$ such that $\SW(\pi) = w$ and $\pi$ witnesses that there is a partition of $V^x$ corresponding to the signature of the record, as described above. We store only one record for every signature -- the one with the highest social welfare. Therefore, in every node $x$, there are at most $n^\Oh{2^{\tw}}$ different records.
	
	Once the computation ends, we check the record in the root node~$r$ and based on the value of $w$, we return the answer; \Yes\ if $w\geq b$ and \No\ otherwise. Moreover, as $G^r=G$, the partition $\pi$ is also an outcome admitting social-welfare~$w$.
	
We conclude the proof sketch by showcasing how the records with the maximum social welfare for each signature can be computed at leaf and introduce nodes; to conclude the proof, it suffices to complete these steps for forget and join nodes, and then argue correctness and the runtime bounds.
	
	\paragraph{Leaf Node.} Leaf nodes are, by definition, empty. Hence, there is only one possible signature $(\emptyset,\emptyset,\emptyset)$, and the value of a record with the highest social welfare is $(\emptyset,0)$.
	
	\paragraph{Introduce Node.} If a node $x$ with a child~$y$ introduces a new agent~$a$, then for every $R_y\in\mathcal{R}_y$ we add multiple records that extend the record $R_y$. The extension of $R_y$ then differs based on its particular signature.
	
	If there is no pair of bag agents $b,b'\in\beta(y)$ such that~${a\in\Spaths(b,b')}$, then $a$ can be placed almost freely to every coalition or form a new singleton coalition.
	In the latter case, the extension of $R_y$ is simple, and the social-welfare remains the same. In the first case, we try to put $a$ in every coalition intersecting bag vertices and create a record for every guess of paths connecting $a$ and other bag vertices of the coalition. We have to drop all guesses where any $(a,b)$-path, $b\in\beta(y)$ and $\Pttn(b) = \Pttn(a)$, introduces a shorter intra-coalitional $(b,b')$-path than the one prescribed by $\Spaths(b,b')$. The social welfare of introduced $R_x$ can be easily recomputed thanks to~$\Dvec$, where we have the number of vertices for all distance vectors to other bag agents in $a$'s coalition, and we guessed paths connecting $a$ with other bag agents of its coalition.
	
	On the other hand, if $a$ is part of some guessed path $\Spaths(b,b')$, $b,b'\in\beta(y)$, then we are forced to put $a$ into~$b$'s coalition. We also reuse the already guessed paths between agents from $a$'s coalition; every time $a$ is part of some $\Spaths(b,b')$, we set $\Spaths(a,b)$ and $\Spaths(a,b')$ to be $(a,b)$-subpath and $(a,b')$-subpath of $\Spaths(b,b')$, respectively. If there is an agent $b$ in $a$'s coalition such that for every $\Spaths(b,b')$ the agent $a$ is not part of this path, we again try all guesses of the $(a,b)$-path. Again, we check that we do not shorten any guessed path. The computation of the social welfare of a newly added record is very similar to the computation of the social welfare in the previous case.	
\end{proof}
\fi
\iflong
\begin{proof}
		\newcommand{\Pttn}{\ensuremath{C}} 	\newcommand{\len}{\ensuremath{\operatorname{len}}}
	\newcommand{\Spaths}{\ensuremath{S}} 	\newcommand{\Dvec}{\ensuremath{T}} 		\begin{notation} In this proof, we use the following notation:
		\begin{itemize}
			\item Let $f$ be a function. By $\mathcal{D}(f)$ we denote the domain of~$f$. \item By $f\downarrow S$, where $S\subseteq \mathcal{D}(f)$, we denote a function $f'$ with $\mathcal{D}(f') = \mathcal{D}(f)\setminus S$ such that $\forall e\in \mathcal{D}(f)\setminus S\colon f'(e) = f(e)$, that is, the elements of $S$ are removed from the domain of $f$.
			\item By $f\uparrow g$, where $g$ is a function such that $\mathcal{D}(g)\cap\mathcal{D}(f) = \emptyset$, we denote a function $f'$ with $\mathcal{D}(f') = \mathcal{D}(f) \cup \mathcal{D}(g)$ and
			\[
				f'(e) = \begin{cases}
					f(e) & \text{if }e\in\mathcal{D}(f),\\
     				g(e) & \text{if }e\in\mathcal{D}(g).
				\end{cases}
			\]
			\item Let $a\in V(G^x)\setminus\beta(x)$ be a forgotten agent, $\beta(x)$ be a bag of a tree-decomposition, and $H$ be a subgraph of $G^x$. We use $\dist_H(a,\beta(x))$ to obtain vector of $a$'s distance to bag vertices in a subgraph $H$.
		\end{itemize}
	\end{notation}
		Since the scoring vector~$\val$ is fixed, there is a constant~$\cdiam$ such that all coalitions are of diameter at most~$\cdiam$.
	
	Our algorithm is based on leaf-to-root dynamic programming along a nice tree-decomposition of the input social network. Therefore, the first step of our algorithm is to compute an optimal tree-decomposition. The algorithm of Bodlaender~\shortcite{Bodlaender96} computes tree-decomposition of optimal width in~${2^\Oh{\tw^3}\cdot n}$ time and this tree-decomposition can be turned into a nice one of the same width in~$\Oh{\tw^2\cdot n}$ time~\cite{CyganFKLMPPS2015}.
	
	With a nice tree-decomposition of optimal width 	in hand, we continue with the description of our dynamic programming algorithm. For each node~$x$ of the nice tree-decomposition from previous step, we define a \emph{signature} as a triplet~$(\Pttn,\Spaths,\Dvec)$, where
		\begin{itemize}
		\item $\Pttn\colon \beta(x)\to[|\beta(x)|]$ is a partition of bag vertices into coalitions,
		\item $\Spaths\colon \beta(x)\times\beta(x)\to \cup_{i=0}^{\cdiam}\binom{(N\setminus V^x)\cup\beta(x)}{i}$ is a set of intra-coalitional shortest paths between each pair of bag vertices that are members of the same coalition, and
		\item $\Dvec\colon [|\beta(x)|]\times\{(d_i)_{i=1}^{|\beta(x)|}\mid d_i\in[\cdiam]\cup\{-\infty\}\}\to[n]_0$ is a table that stores the numbers of agents with specific distances from bag vertices in the same coalition.
	\end{itemize}
		Given a signature $(\Pttn,\Spaths,\Dvec)$, we say that the signature is \emph{valid}, if all of the following conditions hold:
		\begin{itemize}
		\item for every $a\in\beta(x)$ we have $\Spaths(a,a) = \{a\}$,
		\item for each pair of distinct agents $a,b\in\beta(x)$ such that $\Pttn(a) = \Pttn(b)$, it holds that $\Spaths(a,b)$ is a $a,b$-path in $(G\setminus V^x) \cap \{a,b\}$ of length at most $\cdiam$ or $\Spaths(a,b)=\emptyset$ if the shortest path between $a,b$ is (i) trough the past vertices, or (ii) uses any other bag vertices,
		\item for each pair of distinct agents $a,b\in\beta(x)$ such that $\Pttn(a) = \Pttn(b)$ we have $\Spaths(a,b) = \Spaths(b,a)$,
		\item for each quadruplet of distinct agents $a,b,a',b'\in\beta(x)$ such that $\Pttn(a) = \Pttn(b)$, $\Pttn(a') = \Pttn(b')$, and $\Pttn(a) \not= \Pttn(a')$, we have ${\Spaths(a,b)\cap \Spaths(a',b') = \emptyset}$, and
		\item for every $i\in[|\beta(x)|]$ and $\vec{d}\in\{(d_i)_{i=1}^{|\beta(x)|}\mid d_i\in[\cdiam]\cup\{-\infty\}\}$ we have $\Dvec(i,\vec{d}) \geq 0$ only if the positive elements of~$\vec{d}$ corresponds to agents in $\beta(x)$ that are all part of the same coalition according to $\Pttn$; otherwise~$\Dvec(i,\vec{d})$ has to be $0$. 
	\end{itemize}

	Now, we bound the number of possible signatures for every node~$x$ of the tree-decomposition.
	
	\begin{claim}\label{lem:tw_dp_tbl_size}
		For every node~$x$ the number of valid signatures is at most~${n^\Oh{2^{\tw}}}$.
	\end{claim}
	\begin{claimproof}
		Since $\beta(x)$ is a separator in~$G$, there are at most~${\tw+1}$ different coalitions intersecting~$\beta(x)$, therefore there are at most~${(\tw+1)^{\tw+1}}$ different partitions of the bag vertices. 
		
		Recall that~$\val$ is fixed and, hence, the diameter of every coalition is bounded by a constant~$\cdiam$. Therefore, there are~${n^\Oh{\cdiam} = n^\Oh{1}}$ possible paths for each pair of agents and thus for every partition of bag vertices at most~${n^\Oh{\tw^2}}$ combinations of paths in total.
		
		Finally, the vector table~$\Dvec$ is a two-dimensional matrix with a particular coalition as the row dimension and distance vectors as the column dimension. Hence, we have~${\tw \cdot 2^\Oh{\tw}}$ different cells and every cell can have at most~$n$ different values. That is, there are at most~${n^\Oh{\tw\cdot2^{\tw}}}$ possible values of~$\Dvec$.
		
		Overall, we have~${\tw^\Oh{\tw}\cdot n^\Oh{\tw^2}\cdot n^\Oh{\tw2^{\tw}} = n^\Oh{2^{\tw}}}$ valid signatures for every node~$x$, finishing the proof.
	\end{claimproof}
	
	For every valid signature, we store a single \emph{partial solution} or \emph{record} $R=(\Pttn,\Spaths,\Dvec,\Pi^x,w)$ such that
	\begin{itemize}
		\item $\Pi^x$ is a partition of $V^x$ and for every $a,b\in\beta(x)$ such that $\Pttn(a)=\Pttn(b)$ it holds that $\Pi^x_a = \Pi^x_b$,		\item for every coalition $P\in\Pi^x$ and every $a,b\in \beta(x) \cap P$ we have $\dist_{G[P\cup\Spaths(a,b)]}(a,b) = \len(\Spaths(a,b))$, where $\len(\Spaths(a,b))$ is the number of agents in the path $\Spaths(a,b)$,
		\item for every $a \in \beta(x)$ and every possible distance vector $\vec{d}$ we have $\Dvec(\Pttn(a),\vec{d}) = |\{b\mid \Pi^x_b = \Pi^x_a \land \dist_{\Pi^x_a}(b,\beta(x)) = \vec{d}\}|$,
		\item $\SW^*(\Pi^x) = w$, where $\SW^*(\Pi^x) = \sum_{a\in N} \sum_{b\in \Pi^x_a\setminus\{a\}} \val(\dist_{\Pi^x_a\cup\Spaths(a,b)}(a,b))$, and
		\item there is no other partition $\Pi^{'x}$ with the same properties and higher social welfare.
	\end{itemize}
	We say that the record $R$ \emph{realize} the signature $(\Pttn,\Spaths,\Dvec)$. If such a record cannot exists, we simply mark this signature as invalid. In every node~$x$, we store only records for valid signatures. Recall that the number of valid signatures is at most~$n^\Oh{2^{\tw}}$ and so is the number of records.
	
	Observe that once we compute the table for the root node~$r$, whose bag is empty by definition, there is only one valid signature $(\emptyset,\emptyset,\emptyset)$. Based on the value $w$ of the record~$(\emptyset,\emptyset,\emptyset,\Pi^x,w)$, we can return \YES{} if $w \geq b$, and \No{} otherwise. Moreover, as $G^r = G$, we have that $\Pi^x$ is a partition of $N$ with social-welfare exactly~$w$.
	
	Next, we describe how to compute, in a leaf-to-root fashion, records for all valid signatures for every node of the tree-decomposition. We describe the computation separately for each node type. Whenever we compute a record for some signature, we first check whether there is a record with the same signature stored in a node $x$. If so, we leave only the record with the highest social welfare.
	
	\paragraph{Leaf Node.} A bag of a leaf node $x$ is by definition empty. Therefore, there is exactly one valid signature $(\emptyset,\emptyset,\emptyset)$, and the corresponding record $R_x$ is $(\emptyset,\emptyset,\emptyset,\emptyset,0)$. We slightly abused the notation here to capture that $\Pttn$, $\Spaths$, and $\Dvec$ have empty domains.
	
	\paragraph{Introduce Node.} Let an introduce node be~$x$, its child be~$y$, and the introduced agent be~$a$. We add to~$x$ multiple records derived from every record stored in~$y$. Let $R_y=(\Pttn^y,\Spaths^y,\Dvec^y,\Pi^y,w^y)$ be a fixed record stored in~$\mathcal{R}_y$. Based on the appearance of~$a$ on any guessed shortest path $\Spaths(b,b')$, for any pair of $b,b'\in\beta(y)$, the particular procedure of introducing new records in~$x$ differs.
	
	First, suppose that there is no $b,b'\in\beta(x)$ with $\Pttn^y(b) = \Pttn^y(b')$ such that $a\in \Spaths^y(b,b)$. Then, we try to put $a$ in every possible coalition intersecting the bag and to create a new coalition containing only $a$. For the latter case, we create a single record with $\Pttn = \Pttn^y\uparrow\{(a\mapsto|\beta(x)|)\}$, $\Spaths = \Spaths^y\uparrow\{((a,a)\mapsto\{a\})\}$, for every $i\in[|\beta(x)|]$ and every $\vec{d}=(d_1,\ldots,d_{|\beta(x)|})$, where $d_j\in[\cdiam]\cup\{-\infty\}$, we set
	\[
		\Dvec(i,\vec{d}) = \begin{cases}
			\Dvec^y(i,(d_1,\ldots,d_{|\beta(y)|})) & \text{if } i\not= |\beta(x)|\land d_{|\beta(x)|} = -\infty,\\
			0 & \text{otherwise,}
		\end{cases}
	\]
	and with $\Pi^x = \Pi^y \cup \{\{a\}\}$ and $w = w^y$.
	
	For the first case, let $P\in[|\beta(y)|]$ be the number of a non-empty partition we are trying to add $a$ to and $\Pi_P$ be a corresponding coalition in $\Pi^y$. For this partition, we try to create multiple records, check validity of their signatures, and if a record is valid, we add it to $\mathcal{R}_x$. In particular, we try all possible functions $\Spaths'\colon \{a\}\times\{b\mid b\in\Pttn^{y}{}^{-1}(P)\}\times\bigcup_{i=0}^\cdiam\binom{(N\setminus V^x)\cup\Pttn^{y}{}^{-1}(P)}{i}$ that satisfies validity conditions given in the definition of the signature. Let $S'$ be such a fixed function. We create a record $R_x$ with $\Pttn = \Pttn^y \uparrow \{(a \mapsto P)\}$, $\Spaths = \Spaths_y \uparrow S'$,
	\[
		\Dvec(i,\vec{d}) = \begin{cases}
			0  \text{~~~~~~~~~~~~~~~~~~~~~~~~~~if } i = P \text{ and } d_{|\beta(x)|} \not=\\
			 ~~~~~~~~~~~~~~~~~~~~~~~~~~\min\limits_{{b\in\Pi_P\cap\beta(x)}}\{\dist_{\Pi_P\cup\Spaths(a,b)}(a,b)\}\\
			0  \text{~~~~~~~~~~~~~~~~~~~~~~~~~~if } i \not= P \text{ and } d_{|\beta(x)|} \not= -\infty\\
			\Dvec_y\left(i,(d_1,\ldots,d_{|\beta(y)|})\right)  \text{~~~~otherwise,}
		\end{cases}
	\]
	$\Pi^x = \left\{Q\in\Pi^y \mid Q \not= \Pi_P\right\}\cup\left\{\Pi_p\cup\{a\}\right\}$, and
	\begin{multline*}
	w = w^y + \sum_{b\in\Pi_p\cap\beta(x)} \val\left(\dist_{\Pi_p \cup \Spaths(a,b)}(a,b)\right)\\
			+ \sum_{\vec{d}\in\{(d_i)_{i=1}^{|\beta(x)|}\}} 2\cdot\Dvec(P,\vec{d})\cdot\val\left(d_{|\beta(x)|}\right).
	\end{multline*}

	If there exists a pair of agents $b,b'\in\beta(x)\setminus\{a\}$ such that $a\in\Spaths^y(b,b')$, then the computation is the same as in the case of putting $a$ into nonempty partition. There is only one important difference; we have to place $a$ in the coalition $\Pi^y_b$ prescribed by the $\Spaths(b,b')$. 

	Before we finally add the created $R_x$ into $\mathcal{R}_x$, we check whether the signature of $R_x$ meets all validity conditions.
	
	\paragraph{Join Node.} Let $x$ be a join node with two children $y$ and $z$ such that $\beta(x) = \beta(y) = \beta(z)$. Let $R_y=(\Pttn^y,\Spaths^y,\Dvec^y,\Pi^y,w^y)$ be a record in $y$ and $R_z=(\Pttn^z,\Spaths^z,\Dvec^z,\Pi^z,w_z)$ be a record in $z$. We say that the records $R_y$ and $R_z$ are \emph{compatible} if $\Pttn^y = \Pttn^z$ and $\Spaths^y = \Spaths^z$. For every pair of compatible records, we add a record $R_x = (\Pttn=\Pttn^y,\Spaths=\Spaths^y,\Dvec,\Pi^x,w)$, where
		\begin{multline*}
		\forall i\in|\beta(x)|\colon\forall \vec{d}\in\{(d_i)_{i=1}^{|\beta(x)|}\mid d_i \in [\cdiam]\}\colon\\ \Dvec'(i,\vec{d}) = \Dvec^y(i,\vec{d}) + \Dvec^z(i,\vec{d}),
	\end{multline*}
		the new partial coalition structure is defined as
		\begin{multline*}
		\Pi^x = \left\{ P \cup Q \mid P\in\Pi^y \land Q\in\Pi^z \land P\cap\beta(x) = Q\cap\beta(x) \right\} \\
		\cup \left\{P \mid P\in\Pi^y\cup\Pi^z \land P\cap\beta(x) = \emptyset\right\}
	\end{multline*}
		and the new social welfare of $R_x$ is given by the following formula:
		\begin{multline*}
		w = w^y + w_z \\
		- \sum_{\substack{a,b\in\beta(x)\\ \Pttn(a)=\Pttn(b)}} \val\left(\dist_{\Pi^x_a\cup\{\Spaths(c,c')\mid c,c'\in\Pi^x_a\cap\beta(x)\}}(a,b)\right)\\
		+ 
		\sum_{i\in[|\beta(x)|]}
		\sum_{\vec{d}^y,\vec{d}^z\in\{(d_i)_{i=1}^{|\beta(x)|}\}} \Dvec^y(i,\vec{d}^y)\cdot\Dvec^z(i,\vec{d}^z)\cdot \val(\operatorname{vdist}(\vec{d}^y,\vec{d}^z)),
	\end{multline*}
	where $\operatorname{vdist}(\vec{d}^y,\vec{d}^z)$ computes the shortest distance between agents with these two distance vectors. Note that the shortest path between two sets of forgotten agents can only be through~$\beta(x)$ and, therefore, we can find the length of a shortest path by element-wise comparison of distances to the vertices of the bag and selecting the minimal one, that is, $\operatorname{vdist}(\vec{d}^y,\vec{d}^z) = \min\{d_i^y + d_i^z\mid i\in|\beta(x)|\}$. We then add the record created in this way to $\mathcal{R}_x$.
	
	\paragraph{Forget Node.} Let~$x$ be a forget node and~$y$ be its child node with~${\beta(x) = \beta(y) \setminus \{a\}}$, where~${a\in\beta(y)}$. For every record~$R_y={(\Pttn^y,\Spaths^y,\Dvec^y,\Pi^y,w^y)}$ in $\mathcal{R}_y$, we add to $\mathcal{R}_x$ a new record $R_x=(\Pttn,\Spaths,\Dvec,\Pi^x,w)$. Intuitively, we remove~$a$ from its current coalition and update the distance-vectors table $\Dvec$ to not forget that $a$ was a member of $\Pttn(a)$. The social welfare of this record remains unchanged as it was already counted when $a$ was introduced.
	
	Now, we show how the computation works formally. We distinguish two cases based on the size of $\Pttn^y(a)$. Without loss of generality, we assume that $\Pttn^y(a) = 1$ and that in every distance-vector, the element corresponding to the distance to~$a$ is the last one.
	
	First, suppose that $a$ forms a singleton coalition in $R_y$. It follows that either $a$ is in a singleton coalition, or $a$ is the last active member of this coalition. In any case, it is safe to just drop $a$ and appropriate parts of the distance vectors table. Moreover, $a$'s utility was already computed earlier and it will never change. Hence, we set ${\Pttn = \Pttn^y\downarrow\{a\}}$, $\Spaths = \Spaths^y\downarrow\{(a,a)\}$, $\forall i\in[2,|\beta(y)|]$ and $\forall\vec{d}\in\{(d_j)_{j=1}^{|\beta(x)|}\mid d_j\in[\cdiam]\cup\{-\infty\}\}$ we have $\Dvec(i,\vec{d}) = \Dvec^y(i,(d_1,\ldots,d_{|\beta(x)|},-\infty))$, $\Pi^x=\Pi^y$, $w=w^y$, and add $R_x$ to $\mathcal{R}_x$.
	
	If the size of the $a$'s coalition is greater than $1$, we have to be a little more careful; in particular, we need to keep note that $a$ was a member of its coalition. We again set $\Pttn = \Pttn^y\downarrow \{a\}$. Next, since $R_y$ is valid signature, $a$ is not part of any $\Spaths^y(b,b')$, where $\Pttn^y(b)=\Pttn^y(b')=\Pttn^y(a)$. Therefore, we set $\Spaths = \Spaths^y\downarrow\{(a,a)\}$. For every $\vec{d}=(d_1,\ldots,d_{|\beta(x)|})$, where $d_i\in[\cdiam]\cup\{-\infty\}$, we compute the distance-vectors table $\Dvec$ as follows
		\[
	\Dvec(i,\vec{d}) = \begin{cases}
		\sum_{k=1}^\cdiam\Dvec^y\left(i,(d_1,\ldots,d_{\beta(x)},k)\right) & \text{if } i = \Pttn^y(a) \\
		\Dvec^y\left(i,(d_1,\ldots,d_{|\beta(x)|},-\infty)\right) & \text{if } i\not=\Pttn^y(a)
	\end{cases}
	\]
		We also have to note that $a$ was a member of coalition $\Pttn^y(a)$. Let $\vec{d}^a = (d_1^a,\ldots,d_{|\beta(x)|}^a)$ be a vector of distances to other vertices of $\beta(y)$ such that $d_i^a = \dist_{\Pi^x_a\cup\{\Spaths(b,b')\mid b,b'\in\Pi^x_a\}}(a,i)$ We increase the value of $\Dvec(1,\vec{d}^a)$ by one. Finally, we set $\Pi^x=\Pi^y$, $w=w^y$, and add $R_x$ to $\mathcal{R}_x$.
\end{proof}
\fi

Having established \Cref{thm:tw}, we now turn to the considered stability requirements.

\iflong
\begin{theorem}
	\fi
	\ifshort
	\begin{theorem}[$\star$]
		\fi
\label{thm:twstable}
	For every fixed scoring vector \val, the \DGIR and \DGNS problems are in \XP when parameterized by the treewidth of the social network $G$.
\end{theorem}
\ifshort
\begin{proof}[Proof Sketch]
		\newcommand{\Pttn}{\ensuremath{C}} 	\newcommand{\len}{\ensuremath{\operatorname{len}}}
	\newcommand{\Spaths}{\ensuremath{S}} 	\newcommand{\Dvec}{\ensuremath{T}} 		In this proof, we build upon the algorithm from the proof of \Cref{thm:tw}. We start with the case of \DGIR and then we show what we need to change to deal also with \DGNS.
	
	Individual rationality requires that there is no agent with negative utility. To be able to ensure that we do not create outcomes with such deviations, we extend our signatures. Recall that a \emph{signature} is a triplet~$(\Pttn,\Spaths,\Dvec)$, where~$\Pttn$ is partition of bag vertices into coalitions,~$\Spaths$ is a function assigning each pair of vertices that are part of the same coalition according to~$\Pttn$ the shortest path connecting them inside their coalition, and~$\Dvec$ is two dimensional table storing for each coalition~$P$ and each distance vector~$\vec{d}$ the number of forgotten agents with distances~$\vec{d}$ to bag vertices in~$P$. This time, in each table $\Dvec$, we store not only the number of forgotten agents with particular distances, but also the worst utility over all agents with the same distance vector. We call this agent \emph{critical agent}. Observe that each cell can take on at most $n\cdot n\cdot \val(1) = \Oh{n^2}$ values and, thus, there are still at most $n^\Oh{2^{\tw}}$ different tables $\Dvec$. 
	
	The computation is then very similar to the algorithm from \Cref{thm:tw}, in each node, we only need to take care of correct recalculation of $\Dvec$'s values. For leaf nodes, the computation stays same. In every other type of node, before a newly created record is added to $\mathcal{R}_x$, we first recompute utilities of critical agents, and then, if this introduces a negative utility for some critical agent, we skip this record. It is not hard to observe that once the utility of a critical agent is negative, it cannot be improved by introducing new agent. Moreover, every recalculation can be implemented efficiently thanks to the fact that every bag $\beta(x)$ is a separator in $G$.
\end{proof}
\fi
\iflong
\begin{proof}
		\newcommand{\Pttn}{\ensuremath{C}} 	\newcommand{\len}{\ensuremath{\operatorname{len}}}
	\newcommand{\Spaths}{\ensuremath{S}} 	\newcommand{\Dvec}{\ensuremath{T}} 		In this proof, we build upon the algorithm from the proof of \Cref{thm:tw}. We start with the case of \DGIR and then we show what we need to change to deal also with \DGNS.
	
	\textbf{Individual rationality} requires that there is no agent with negative utility. To be able to ensure that we do not create outcomes with IR-deviations, we extend our signatures. Again, for a node $x$ of the tree-decomposition, we define a \emph{signature} as a triplet~$(\Pttn,\Spaths,\Dvec)$, where
	\begin{itemize}
		\item $\Pttn\colon \beta(x)\to [|\beta(x)|]$ is a partition of bag vertices into coalitions,
		\item $\Spaths\colon \beta(x)\times\beta(x)\to \cup_{i=1}^{\cdiam}\binom{N}{i}$ is a function assigning each pair of vertices that are part of the same coalition according to~$\Pttn$ a shortest path connecting them inside their coalition, and
		\item $\Dvec\colon [|\beta(x)|]\times\{(d_i)_{i=1}^{|\beta(x)|}\mid d_i\in[\cdiam]\cup\{-\infty\}\}\to[n]_0\times([\val(1)\cdot n]_0\cup\{\infty\})$ is a two dimensional table storing, for each coalition~$P$ and each distance vector~$\vec{d}$, the number of forgotten agents with distances~$\vec{d}$ to bag vertices in~$P$ and, additionally, also the worst utility over all agents with the distance vector $\vec{d}$. We call this agent a \emph{critical agent}.
	\end{itemize} 
	It follows that the only difference is the vector table~$\Dvec$ that additionally stores also utilities of critical agents. Observe that for each cell of $\Dvec$, there are at most $n\cdot n\cdot \val(1) = \Oh{n^2}$ possible different values and, thus, there are still at most $n^\Oh{2^{\tw}}$ different tables $\Dvec$. 
	
	We also extend our requirements on signature validity. In particular, we additionally require that for every $i\in[\beta(x)]$ and every $\vec{d}$ we have $\Dvec(i,\vec{d})[2] \not= \infty$ if and only if $\Dvec(i,\vec{d})[1] > 0$.
	
	Moreover, every record $R$ also has to additionally satisfy that $\forall\Pi_i^x\in\Pi^x$, such that $\Pi_i^x \cap \beta(x) = \emptyset$, there is no agent admitting IR-deviation in $\Pi_i^x$.
	
	The computation is then very similar to the algorithm from \Cref{thm:tw}, in each node, we only need to take care of the correct recalculation of the utility of critical agents stored in $\Dvec$. In what follows, we describe how this recomputation can be done in each node type. A high-level idea is that before a newly created record is added to $\mathcal{R}_x$, we recompute utilities of critical agents. It is not hard to observe that it is enough to store the utility of a single critical agent for every distance vector, as introduction of a new agent affects all agents in this distance equally. Moreover, every recalculation can be implemented efficiently. Both of these properties hold due to the fact that every bag $\beta(x)$ is a separator in $G$.
	
	In order to keep the description as condensed, as possible, we only describe how we recompute the utility of critical agents, because the computation of all the remaining components of records remains the same. We perform this recomputation as a final step before adding $R_x$ to $\mathcal{R}_x$. We denote as~$\Dvec_1$ the first element stored in each cell of $\Dvec$ (that is, the number of past agents with specific distances to the bag vertices), and~$\Dvec_2$ the second element -- the utility of a critical agent.
	
	\paragraph{Leaf Node.} In leaf nodes, the only valid signature is $(\emptyset,\emptyset,\emptyset)$, where $\emptyset$ represents that $\Pttn$, $\Spaths$, and $\Dvec$ have empty domains. Hence, there are no critical agents and the computation remains completely the same.
	
	\paragraph{Introduce Node.} Let $x$ be an introduce node, $y$ be its child, and $a$ be the introduced agent. For every record $R_x$ that would be added to $\mathcal{R}_x$, we recompute utilities of critical agents. This can be achieved as follows.
	
	If $|\Pi^x_a| = 1$, then there is nothing to do, as there are no critical agent in $a$'s coalition. Hence, we focus on the case when $|\Pi^x_a| > 1$. For every $i\in[|\beta(x)|]$ and every $\vec{d}=(d_1,\ldots,d_{|\beta(x)|})$, where w.l.o.g $d_{|\beta(x)|}$ represents distance to $a$, we set
	\[
	\Dvec_2^x(i,\vec{d}) = \begin{cases}
		\Dvec_2^y\left(i,(d_1,\ldots,d_{|\beta(x)-1|})\right) + \val\left(d_{|\beta(x)|}\right) \\
		 \text{~~~~~~~~~~~~~~~~~~~~~~~~~~~if }i = \Pttn(a) \land d_{|\beta(x)|} \not= -\infty \\
		\Dvec_2^y\left(i,(d_1,\ldots,d_{|\beta(x)-1|})\right) ~~~~~~~~ \text{otherwise}.
	\end{cases}
	\]
	
	\paragraph{Join Node.} 
	For every $i\in[|\beta(x)|]$ and every $\vec{d} \in \{(d_j)_{j=1}^{|\beta(x)|}\mid d_j \in [\cdiam]\cup\{-\infty\}\}$ we set
		\begin{multline*}
		\Dvec_2(i,\vec{d}) =\min_{w\in{y,z}}\Big\{
			\Dvec_2^w(i,\vec{d})\, + \\
			\sum_{\vec{d}'} \Dvec_1^{\{x,y\}\setminus w}(i,\vec{d}')\cdot\val\left(\min\{d_j+d^{'}_j\mid j\in[|\beta(x)|]\}\right)
		\Big\},
	\end{multline*}
	
	\paragraph{Forget Node.} The computation in the forget node need only a cosmetic change. We again distinguish two cases based on the size of $a$'s partition. 
	
	First, let the size of $a$'s partition in $R_y$ be one. Then, we check whether all agents in $\Pi^x_a$, including $a$, have non-negative utility. If it is the case, we add $R_x$ to $\mathcal{R}_x$. Otherwise, we skip $R_x$ as $a$ is the last element in this coalition and the utilities can never improve.
	
	Now, let the size of $a$'s partition in $R_y$ be greater than one and $\vec{d}^a = (d_1^a,\ldots,d_{|\beta(x)|}^a)$ be a vector of $a$'s distances to bag vertices in $G[\Pi_{a}^x]$. For every $i\in[\beta(x)]$ and every $\vec{d}=(d_1,\ldots,d_{|\beta(x)|})$ we set $\Dvec_2(i,\vec{d}) = \min\{\Dvec_2^y(i,(d_1,\ldots,d_{|\beta(x)|},j))\mid j \in[\cdiam]\cup\{-\infty\}\}$. As the final step, we update $\Dvec_2(\Pttn^y(a),\vec{d}^a)$ to be $\min\{\Dvec_2(\Pttn^y(a),\vec{d}^a),\util_a(\Pi_a^x)\}$.
	
	\bigskip

	\textbf{Nash stability} requires that no agent has the incentive to form a singleton coalition, as it was in the case of individual rationality, and, additionally, prefers its current coalition over joining any other coalition. We know how to deal with the first type of deviation from the above algorithm, and now we show how to extend signatures to be able to also deal with the second kind.
	
	Before we describe the new signatures, we observe a few properties of coalitions that give the intuition behind the signature definition. From an agent $a$ perspective, we recognize three types of coalitions. A \emph{past} coalition is a coalition such that all of its members are already forgotten in the nice tree decomposition. An \emph{active} coalition is a coalition such that at least one introduced but not forgotten agent appears in the same bag as $a$, and \emph{future} coalition is a coalition whose agents are not yet introduced and never appear in the same bag as $a$.
	
	Let $a$ be an agent and $x$ be a node of the tree decomposition. It is easy to see that agent's $a$ utility in any past coalition will always stay the same, as past coalitions remain the same for the rest of the algorithm. Hence, we can only store the $a$'s utility in the most preferred past coalition in our distance-vectors table. This is different for the active coalitions, as, most likely, there will be agents added to active coalitions in the future. Therefore, we will store (and recompute) the agent's utility in every currently active coalition at every step of the algorithm. Recall that there are at most $\beta(x) = \Oh{\tw}$ active coalitions in every node. For the future coalitions, it holds that $a$ will never deviate to them; let $b$ be any agent of some future coalition since $a$ and $b$ are never part of the same bag, and $b$ is not yet introduced, $a$ can never deviate to $\Pi_b$ as $b$ is not a neighbor of $a$ in $G$ and, thus, $a$'s utility in $\Pi_b$ would be $-\infty$. That is, we do not need to store any information about future coalitions.
	
	\paragraph{Leaf Node.} The computation in leaf nodes, again, remains the same. The only valid signature is $(\emptyset,\emptyset,\emptyset)$, where $\emptyset$ denotes that the domain of respective functions is empty. It follows that we do not need to compute anything for an empty graph.
	
	\paragraph{Introduce Node.} The introduction of a new agent $a$ can have two different outcomes -- either $a$ creates a new singleton coalition or joins an existing one.
	
	In the first case, for every neighbor of $a$ in the bag $\beta(x)$, we need to note that it can have the incentive to deviate to $a$'s coalition together with its current utility in $\Pi_a$. The same we do for $a$ and all the remaining active coalitions containing at least one neighbor of the agent $a$.
	
	On the other hand, if $a$ joins an existing coalition, we need to recompute the utility of all critical agents and remember $a$'s utility in other active coalitions, as $a$ can also admit NS-deviation in the future.
	
	\paragraph{Join Node.} The only additional thing we need to do in a join node is to correctly recompute the utilities of critical agents and their current utilities.
	
	\paragraph{Forget Node.} In a forget node, we again recognize two distinct cases. Either $a$ is the last member of its coalition and, hence, the $\Pi_a$ is becoming completely past, or there are still some agents of $\Pi_a$ in the bag $\beta(x)$.
	
	In the first case, we, for every critical agent, recompute its best utility in past coalitions as the maximum of the current best utility in past coalitions and utility in $\Pi_a$. Moreover, as we forget the last member of $\Pi_a$, we can also recognize whether there is an agent with NS-deviation in $\Pi_a$. If it is the case, we will drop this record. Otherwise, we add it to $\mathcal{R}_x$.
	
	In the latter case, we note in the distance vectors table that the agent $a$ was a member of coalition $\Pi_a$ and also store its utilities in a) its most preferred past coalition and b) all other active coalitions.
\end{proof}
\fi

\subsection{Fixed-Parameter Tractability}
A natural follow-up question to Theorems~\ref{thm:tw} and~\ref{thm:twstable} is whether one can improve these results to fixed-parameter algorithms. As our final contribution, we show that this is possible at least when dealing with simple scoring vectors, or on networks with stronger structural restrictions. To obtain both of these results, we first show that to obtain fixed-parameter tractability it suffices to have a bound on the size of the largest coalition in a solution (i.e., a welfare-optimal outcome).

\iflong
\begin{theorem}
\fi
\ifshort
\begin{theorem}[$\star$]
\fi
\label{thm:tw+coal-sz}
	For every fixed scoring vector~$\val$,  \DGWF, \DGIR, \DGNS  are \FPT w.r.t.\ the treewidth of the network plus the maximum coalition size.
\end{theorem}

\ifshort
\begin{proof}[Proof Sketch]
Similar to the previous ones, we design a dynamic programming (DP) on a nice tree decomposition, albeit the procedure and records are completely different. 

Given a subset of agents $X \subseteq N$, let $\Pi = (\pi_1,\pi_2, \dots, \pi_\ell)$ be a partition of a set containing $X$ and some ``anonymous'' agents. We use \emph{$\topo(\Pi)$} to denote a set of graph topologies on $\pi_1, \pi_2, \dots, \pi_\ell$ given $X$. That is, $\topo(\Pi) = \{ \topo(\pi_1), \dots , \topo(\pi_\ell)\}$ where $\topo(\pi_i)$ is some graph on $|\pi_i|$ agents, namely $\pi_i \cap X$ and $|\pi_i \setminus X|$ ``anonymous'' agents, for each $i \in [\ell]$. 
The maximum coalition size of any welfare maximizing partition is denoted by $\sz$.
Table, {\sf M}, contains an entry {\sf M}$[x, C, \topo(\Pi)]$ for every node $x$ of the tree decomposition,  each partition $C$ of $\beta(x)$, and each set of graph topologies $\topo(\Pi)$ given $\beta(x)$ where $\Pi$ is a partition of at most $\sz\cdot\tw$ agents. An entry of {\sf M} stores the maximum welfare in $G^x$ under the condition that the partition into coalitions satisfies the following properties.
	Recall that for a partition $P$ of agents and an agent $a$, we use $P_a$ to denote the coalition agent $a$ is part of in $P$. \begin{enumerate}
\item \emph{$C$ and $\Pi$ are consistent}, i.e., the partition of the bag agents $\beta(x)$ in $G^x$ is denoted by $C$ and $C_a = \Pi_a \cap \beta(x)$ for each agent $a \in \beta(x)$.
\item The coalition of agent $a \!\in \!\beta(x)$ in the graph $G^x$ is $\Pi_a$.
\item \emph{$\topo(\Pi)$ is consistent with $G^x$} i.e., the subgraph of $G^x$ induced on the agents in coalition of $a$ is $\topo(\Pi_a)$, i.e., $G^x[\Pi_a] = \topo(\Pi_a)$.
\end{enumerate}	
	 
Observe that we do not store $\Pi$. We only store the topology of $\Pi$ which is a graph on at most $\sz\cdot \tw$ agents. 
	We say an entry of {\sf M}$[x,C, \topo(\Pi)]$ is \emph{valid} if it holds that \begin{enumerate}
	\item \emph{$C$ and $\Pi$ are consistent}, i.e., $C_a = \Pi_a \cap \beta(x)$ for each agent $a\in \beta(x)$,
	\item Either $C_a = C_b$, or $C_a \cap C_b = \emptyset$ for each pair of agents $a,b \in \beta(x)$,
	\item \emph{$\topo(\Pi)$ is consistent with $G^x$ in $\beta(x)$}, i.e., for each pair of agents $a,b \in \beta(x)$ such that $\Pi_{a} = \Pi_{b}$, there is an edge $(a,b) \in \topo(\Pi_{a})$ if and only if $(a,b)$ is an edge in $G^x$.
	\end{enumerate}
	
	Once the table is computed correctly, the solution is given by the value stored in  {\sf M}$[r,C, \topo(\Pi)]$ where $C$ is empty partition and $\topo(\Pi)$ is empty. Roughly speaking, the basis corresponds to leaves (whose bags are empty), and are initialized to store $0$. For each entry that is not valid we store $-\infty$. To complete the proof, it now suffices to describe the computation of the records at each of the three non-trivial types of nodes in the decomposition and prove correctness.
\end{proof}
\fi

\iflong
\begin{proof}
\newcommand{\res}[1]{\ensuremath{_{#1}}}
\newcommand{\agents}[1]{\ensuremath{\cup{#1}}}

\textbf{Proof for \DGWF.} We begin by defining some notation. Let $\tw$ denote the treewidth of $G$ and $\sz$ denote the size of the maximum coalition in any welfare maximizing partition of $G$. 
Let $\Pi$ be a partition of a set of agents $N' \subseteq N$ and $\Pi'$ be a partition of $N'' \subseteq N'$. Then we say \emph{$\Pi$ complies with $\Pi'$} if  for each agent $a \in N''$, $\Pi'_a = \Pi_a$. We denote it as $\Pi\res{N'} = \Pi'$. Moreover, we use $\agents{\Pi}$ to denote the agents $N'$.
Let $\Pi = (\pi_1,\pi_2, \dots, \pi_\ell)$ be a partition of agents $N'$.
Given a subset of agents $X \subseteq N'$, we use \emph{$\topo(\Pi)$} to denote a set of graph topologies on $\pi_1, \pi_2, \dots, \pi_\ell$ for agents $X$. That is, $\topo(\Pi) = \{ \topo(\pi_1), \dots , \topo(\pi_\ell)\}$ where $\topo(\pi_i)$ is some graph on $|\pi_i|$ agents, namely $\pi_i \cap X$ and $|\pi_i \setminus X|$ ``anonymous'' agents, for each $i \in [\ell]$. 
Additionally, given $\topo(\Pi)$ and a coalition $\pi \in \Pi$, we use $\SW(\topo(\pi))$ to denote the contribution of the coalition $\pi \in \Pi$ to welfare according to the topology $\topo(\Pi)$, i.e.,
$$\SW(\topo(\pi)) = \sum_{a \in \pi} \sum_{b \in \pi \setminus \{a\}} \val(\dist_{\topo(\pi)} (a,b))$$
where $\dist_{\topo(\pi)} (a,b)$ is the distance between agents $a$ and $b$ in the coalition $\pi$ according to the topology $\topo(\pi)$.

	Our algorithm is based on dynamic programming (DP) on a nice tree decomposition.  Our DP table, {\sf M}, contains an entry {\sf M}$[x, C, \topo(\Pi)]$ for every node $x$ of the tree decomposition,  each partition $C$ of $\beta(x)$, and each set of graph topologies $\topo(\Pi)$ where $\Pi$ is a partition of at most $\sz\cdot\tw$ agents including $\beta(x)$ and the topologies are defined for set of agents $\beta(x)$. An entry of {\sf M} stores the maximum welfare in $G^x$ under the condition that the partition into coalitions satisfies the following properties.
	Recall that for a partition $P$ of agents and an agent $a$, we use $P_a$ to denote the coalition agent $a$ is part of in $P$. Then, for {\sf M}$[x, C, \topo(\Pi)]$ we consider partitions of \(G^x\) for which
\begin{enumerate}[label=P\arabic*]
\item \emph{$C$ and $\Pi$ are consistent}, i.e., the partition of the bag agents $\beta(x)$ in $G^x$ is denoted by $C$ and $C_a = \Pi_a \cap \beta(x)$ for each $a \in \beta(x)$;
\item the coalition of an agent $a \in \beta(x)$ in the graph $G^x$ is $\Pi_a$. Then $\beta(x) \subseteq \cup\Pi \subseteq V^x$ and $\Pi$ is a partition of $\cup\Pi$;
\item \emph{$\topo(\Pi)$ is consistent with $G^x$} i.e., the subgraph of $G^x$ induced on the agents in coalition of $a$ is $\topo(\Pi_a)$, i.e., $G^x[\Pi_a] = \topo(\Pi_a)$.
\end{enumerate}	
	 
Observe that we do not store $\Pi$. We only store the topology of $\Pi$ which is a graph on at most $\sz\cdot \tw$ agents. 
	We say an entry of {\sf M}$[x,C, \topo(\Pi)]$ is \emph{valid} if it holds that 
	\begin{enumerate}[label=V\arabic*]
	\item \emph{$C$ and $\Pi$ are consistent}, i.e., $C_a = \Pi_a \cap \beta(x)$ for each agent $a\in \beta(x)$,
	\item Either $\Pi_a = \Pi_b$, or $\Pi_a \cap \Pi_b = \emptyset$ for each pair of agents $a,b \in \beta(x)$,
	\item\label{it:welval3} \emph{$\topo(\Pi)$ is consistent with $G^x$ in $\beta(x)$}, i.e., for each pair of agents $a,b \in \beta(x)$ such that $\Pi_{a} = \Pi_{b}$, there is an edge $(a,b) \in \topo(\Pi_{a})$ if and only if $(a,b)$ is an edge in $G^x$.
	\end{enumerate}
	
	Once the table is computed correctly, the solution is given by the value stored in  {\sf M}$[r,C, \topo(\Pi)]$ where $C$ is empty partition and $\topo(\Pi)$ is empty. Roughly speaking, the basis corresponds to leaves (whose bags are empty), and are initialized to store $0$. For each entry that is not valid we store $-\infty$.
	
	\noindent
	\textbf{Introduce Node.}
	Given a node $x$ whose bag contains one more agent, $a$, than its child $y$, 
	we recursively define {\sf M}$[x,C, \topo(\Pi)]$. Let  $C' = C\res{\beta(y)}$ and $\topo(\Pi') = \{\topo(\pi) \mid \pi \in \Pi, \pi \neq \Pi_a\} \cup \topo(\Pi_a \setminus \{a\})$ where $\topo(\Pi_a \setminus \{a\})$ denotes the graph we get by deleting agent $a$ from $\topo(\Pi_a)$. Then, $\Pi' = \Pi\res{\beta(y)}$. Let $\Delta(\topo(\Pi), \topo(\Pi'))$ denote the difference in welfare due to adding $a$ to the coalition $\Pi_a$. Then, we define $\Delta(\topo(\Pi), \topo(\Pi'))= \SW(\topo(\Pi_a)) -\SW(\topo(\Pi_a \setminus \{a\}))$. Then,
	\begin{align*}\label{eq:welinode}\tag{Inode}
{\sf M}[x,C, \topo(\Pi)] 
&= {\sf M}[y,C', \topo(\Pi')] + \Delta(\topo(\Pi), \topo(\Pi')).
\end{align*}

\noindent
\textbf{Forget Node.} Given a node $x$ whose bag contains one fewer agent $a$ than its child $y$, we examine each coalition this agent can be in, and take the one that maximizes the entry of the child. So, we consider any coalition structure $C'$ of $\beta(y)$ such that $C'\res{\beta(x)} =C$.  If $C'_a = \{a\}$, i.e., agent $a$ is in a different coalition in $\beta(y)$ than any agent in $\beta(x)$, then we take maximum over all possible topologies, $\topo(\Pi'_a)$, containing $a$ of size at most $\sz$. Otherwise, $(\emptyset \neq) C'_a \setminus \{a\} \in C$. Then $\topo(\Pi') = \topo(\Pi)$. Hence, we define,

\begin{align*}\label{eq:welfnode}\tag{Fnode}
{\sf M}[x,C, \topo(\Pi)]
&= \max_{C', \topo(\Pi')}{\sf M}[y,C', \topo(\Pi')] 
\end{align*}

\noindent
\textbf{Join Node.} For a node $x$ with two children $y_1$ and $y_2$, we compute the sum of the entries associated with these children (where $C$ is the partition of coalitions in both $y_1$ and $y_2$), from which we modify the welfare of agents in $\beta(x)$ as they are computed in subgraphs $G^{y_1}$ and $G^{y_2}$ and not in $G^x$.

Let $a$ be an agent in $\beta(x)$. Then, $\topo(\Pi^1_a)$ and $\topo(\Pi^2_a)$ denote the two topologies of coalitions of agent $a$ in  $y_1$ and $y_2$, respectively, and  $C_a$ is a subgraph of both $\topo(\Pi^1_a)$ and $\topo(\Pi^2_a)$.
Let $\topo(\Pi^1_a) \cup \topo(\Pi^2_a)$ define the topology we get by identifying the agents of $C_a$ in the two topologies. Since the agents of $\beta(x)$ are not anonymous, so we can identify them.  Moreover, we get $\topo(\Pi^1) \cup \topo(\Pi^2)$ by doing so for each $b \in \beta(x)$.
Let $\Delta(\topo(\Pi), \topo(\Pi^1), \topo(\Pi^2))$ denote the difference between the total utility of the coalitions that contain agents of $\beta(x)$ in $G^x$ and the total utility of these coalitions in $G^{y_1}$ and $G^{y_2}$. Recall that total utility of the agents in a partition $P$ is denoted by $\SW(P)$. That is $\Delta(\topo(\Pi), \topo(\Pi^1), \topo(\Pi^2)) = \SW(\topo(\Pi^1) \cup \topo(\Pi^2))- \SW(\topo(\Pi^1) -\SW(\topo(\Pi^2)$. Then,

\begin{multline*}\label{eq:weljnode}\tag{Jnode}
{\sf M}[x,C, \topo(\Pi)] 
= \max_{\topo(\Pi^1),\topo(\Pi^2)}\{{\sf M}[y_1,C, \topo(\Pi^1)] \\
+  {\sf M}[y_2,C, \topo(\Pi^2)] + \Delta(\topo(\Pi), \topo(\Pi^1), \topo(\Pi^2))\}.
\end{multline*}

Let $x$ be a leaf node. Since, $\beta(x)=\emptyset$, we only have one value ${\sf M}[x,C, \topo(\Pi)]  = 0$. Properties P1-P3 holds due to the validity checks V1-V3.
The proof of correctness of the formulas is given in the next lemma.

\begin{lemma}
 The recurrences \eqref{eq:welinode},\eqref{eq:welfnode}, and \eqref{eq:weljnode} are correct. 
\end{lemma}\label{lem:welfare+coalDPcorrectness}
\begin{proof} 
 We analyze the three nodes Introduce, Forget, and Join separately and prove that ${\sf M}$ stores the maximum welfare. We will prove that each of the recurrence relation holds by showing inequality in both directions.

 \noindent
{\bf Introduce Node}: Let $x$ be a node whose bag contains one more agent, $a$, than its child $y$. Observe that,  $a$ is the new agent introduced in $C_a$ and hence, in $\Pi_a$. We only need to recompute the contribution of $\Pi_a$ in $\beta(x)$ towards welfare. Note that the rest of coalitions are in $C'$ as well and their contribution to welfare remains the same as ${\sf M}[y,C', \topo(\Pi')]$. Since we know the graph $\topo(\Pi_a)$ satisfy condition \eqref{it:welval3} in the initial validity check, clearly, the sum of utility of agents in $\Pi_a$ in $G^x$ is $\SW(\topo(\Pi_a))$. Moreover,  since $\Pi\res{\beta(y)} = \Pi'$ and we get $\topo(\Pi')$ by deleting $a$ from $\topo(\Pi)$, the total utility of agents in $\Pi_a \setminus \{a\}$ in $G^y$ is $\SW(\topo(\Pi_a\setminus\{a\})$. Therefore, the change in total utility is given by $\Delta(\topo(\Pi), \topo(\Pi'))$ since $G^y = G^x -a$.
Therefore, since ${\sf M}[y,C', \topo(\Pi')]$ is maximum welfare in $G^y$ given partition $\Pi'$, we have that  ${\sf M}[y,C', \topo(\Pi')] + \Delta(\topo(\Pi), \topo(\Pi'))$ is at least ${\sf M}[x,C, \topo(\Pi)]$. 

Now since ${\sf M}[x,C, \topo(\Pi)]$ stores the maximum welfare for $\Pi$ in $G^x$ and $G^y$ is a subgraph of $G^x$, we have that any partition in $G^y$ can have total utility at most ${\sf M}[x,C, \topo(\Pi)]$ barring the change in total utility of agents due to removing $a$ from its coalition in $\Pi$. Therefore, ${\sf M}[y,C', \topo(\Pi')] + \Delta(\topo(\Pi), \topo(\Pi')) \leq {\sf M}[x,C, \topo(\Pi)]$. Hence, we prove the formula.

\noindent
{\bf Forget Node}: Let $x$ be a node whose bag contains one fewer agent $a$ than its child $y$.   Observe that, for any partition $C'$ and topology $\topo(\Pi')$ such that $C = C'\res{\beta(x)}$ and $\topo(\Pi)$ is the restriction of $\topo(\Pi')$ to $\beta(x)$, each partition that is considered in the definition of ${\sf M}[x,C', \topo(\Pi')] $ is also considered in the definition of ${\sf M}[x,C, \topo(\Pi)]$. This implies, that for any such $C'$ and $\topo(\Pi')$, we have ${\sf M}[x,C, \topo(\Pi)]  \geq {\sf M}[y,C', \topo(\Pi')]$, that is, ${\sf M}[x,C, \topo(\Pi)] \geq \max_{C', \topo(\Pi')} {\sf M}[y,C', \topo(\Pi')]$.

On the other hand, let $\hat{\Pi}$ be a partition maximizing the welfare in $G^x$ given that the coalitions of agents of $\beta(x)$ is $\Pi$. Let us define $C'$ as follows for an agent $b \in \beta(y)$\[C'_b = 
\hat{\Pi}_b \cap \beta(y) \text{ ~~for each }  b \in \beta(y) \]
and $\topo(\Pi')$ is the same as the graph $G^x[\agents{\Pi}]$. Note that $C'\res{\beta(x)}$ is the same as $C$.
Therefore, $C'$ and $\topo(\Pi')$ are one of the partitions and topologies that are considered in the definition of ${\sf M}[y,C', \topo(\Pi')]$, we have, $\max_{C', \topo(\Pi')} {\sf M}[y,C', \topo(\Pi')] \geq {\sf M}[x,C, \topo(\Pi)]$. Therefore, we prove the equality.

\noindent
{\bf Join Node}: Let $x$ be a join node with children $y_1$ and $y_2$. So, $\beta(x) = \beta(y_1) = \beta(y_2)$. 
 Recall that given $\topo(\Pi)$, the total utility of coalition of $a$ is denoted by $\SW(\topo(\Pi_a))$. Therefore, ${\sf M}[y_1,C, \topo(\Pi^1)] + {\sf M}[y_2,C, \topo(\Pi^2)]$ adds the contributions of partial coalitions of $a$ present in the subgraphs $G^{y_1}$  and $G^{y_2}$. Recall that $G^{y_1} \cap G^{y_2} =\beta(x)$. So we deduct the contributions in the subgraphs, i.e., we deduct $\SW(\topo(\Pi^1) 
 +\SW(\topo(\Pi^2)$ and add the contribution  $\SW(\topo(\Pi^1) \cup \topo(\Pi^2))$ of the (possibly) larger coalitions in $G^x$  to compute the welfare of the entry ${\sf M}[x,C, \topo(\Pi)]$. Hence, we add $\Delta(\topo(\Pi), \topo(\Pi^1), \topo(\Pi^2))$ to compute the welfare for $\topo(\Pi)$ in $G^x$.  Recall ${\sf M}[x,C, \topo(\Pi)]$ stores the maximum welfare of $G^x$. Therefore, ${\sf M}[x,C, \topo(\Pi)] 
\geq \max_{\topo(\Pi^1),\topo(\Pi^2)}\{{\sf M}[y_1,C, \topo(\Pi^1)] 
+  {\sf M}[y_2,C, \topo(\Pi^2)] \\
+ \Delta(\topo(\Pi), \topo(\Pi^1), \topo(\Pi^2))\}$.
 
For the other direction, let $\hat{\Pi}$ be an welfare maximizing partition of $G^x$ such that $\Pi_a = \hat{\Pi}(a)$  and the topology of $G^x[\hat{\Pi}_a]$ is the same as $\topo(\Pi_a)$ for each $a \in \beta(x)$.
Since $G^{y_1}$ and $G^{y_2}$ are subgraphs of $ G^x$, we define $\hat{\Pi}^1$ and $\hat{\Pi}^2$ to be the coalitions obtained by restricting $\hat{\Pi}$ to $G^{y_1}$ and $G^{y_2}$, respectively. Note that the contribution of coalitions containing agents in $\beta(x)$ to welfare in $G^x$ is 
\[ \sum_{a\in \beta(x)} \sum_{b \in \hat{\Pi}_a} \util^{\val}(b,\hat{\Pi})\]
\[=\sum_{a\in \beta(x)} \sum_{b \in \hat{\Pi}_a} \util^{\val}(b,\hat{\Pi}^1 \cup \hat{\Pi}^2)).\]

The last equality follows since $G^x$ is the union of $G^{y_1}$ and $G^{y_2}$.
 However, in $\sum_{b \in \hat{\Pi}_a} \util^{\val}(b,\hat{\Pi}^i)$ we add the contribution of coalition of $a$ only in the subgraph $G^{y_i}$ for each $i \in [2]$. Therefore we deduct $\sum_{b \in \hat{\Pi}_a} \util^{\val}(b,\hat{\Pi}^1) + \sum_{b \in \hat{\Pi}_a} \util^{\val}(b,\hat{\Pi}^2)$ and add the contribution of the coalition of $a$ in $G^x$, yielding $\sum_{b \in \hat{\Pi}_a} \util^{\val}(b,\hat{\Pi}^1 \cup \hat{\Pi}^2)) - \sum_{b \in \hat{\Pi}_a} \util^{\val}(b,\hat{\Pi}^1) - \sum_{b \in \hat{\Pi}_a} \util^{\val}(b,\hat{\Pi}^2)$.
 
 Therefore, ${\sf M}[x,C, \topo(\Pi)]$ is at most ${\sf M}[y_1,C, \topo(\hat{\Pi}^1)] \\
+  {\sf M}[y_2,C, \topo(\hat{\Pi}^2)] + \Delta(\topo(\hat{\Pi}), \topo(\hat{\Pi}^1), \topo(\hat{\Pi}^2))$. Hence,  ${\sf M}[x,C, \topo(\Pi)] 
\leq \max_{\topo(\Pi^1),\topo(\Pi^2)}\{{\sf M}[y_1,C, \topo(\Pi^1)] \\
+  {\sf M}[y_2,C, \topo(\Pi^2)] + \Delta(\topo(\Pi), \topo(\Pi^1), \topo(\Pi^2))\}$.
 Therefore, we get the equality given in the recurrence.

This concludes the description and the proof of correctness of the recursive formulas for computing the values of the table ${\sf M}$.
\end{proof}

There are at most $\tw^\Oh{\tw}$ partitions of $\beta(x)$ for a node $x$ and $\topo(\Pi)$ is a set of at most $\tw$ graph topologies since there are at most $\beta(x)$ coalitions in $\Pi$, each on at most $\sz$ agents since the coalition size are bounded by $\sz$. Then, the number of possible topologies is $(2^{\Oh{\sz^2}})^{\tw}$. Each entry can be computed in time  $(\tw \cdot 2^{\Oh{\sz^2}})^{\Oh{\tw}}$. Each of the validity conditions can be checked in polynomial time. Thus, the algorithm takes time $(\tw \cdot 2^{\Oh{\sz^2}})^{\Oh{\tw}}$. This concludes the proof of the theorem for \DGWF.

Next, we give a proof sketch for \DGIR since it is very similar to \DGWF and we present a complete proof for \DGNS.

\textbf{For \DGIR} the recurrence is the same for Forget node as in \DGWF. For Introduce (and Join node) we do an extra check that if by adding a new agent $a$ (resp. by joining the coalitions in two subgraphs), the utility of some agent is less than zero in $\topo(\Pi)$, then we set ${\sf M}[x, C, \topo(\Pi)] $  to be $-\infty$. The remaining computation remains the same. The correctness follows from previous lemma and the fact that utility of no agents changes at a forget node and if utility of an agent changes at an Introduce or a Join node, then this agent must be in $\Pi$.

\newcommand{\minU}[1]{\ensuremath{\underline{u}_{#1}}}
\newcommand{\maxU}[1]{\ensuremath{\overline{u}_{#1}}}
\newcommand{\extn}[1]{\text{extn-}\ensuremath{#1}}

\newcommand{\argmin}{\operatorname{argmin}}

\textbf{Proof for \DGNS.} We will use a DP table similar to the algorithm for \DGWF with some more book keeping to track possible Nash deviations.
Given a subset of agents $X \subseteq N$, let $\Pi$ be a partition of a set of agents $N'$ containing $X$ and some ``anonymous'' agents. For each agent $a \in N'$, we use $F(a)$ to denote an anonymous agent that is a neighbor of $a$ and refer to it as friend of $a$. We define the set $\extn{\Pi}$ that has the agents of $\Pi$ and the anonymous friend for each agent in $\Pi$, i.e., $\extn{\Pi} = \cup_{a\in N'} F(a) \cup N'$.  Then, define topology of $\extn{\Pi}$, \emph{$\topo(\extn{\Pi})$ given $X$} to be a graph structure on $\extn{\Pi}$ such that $(a, F(a))$ is the only edge incident to $F(a)$ in the topology. 

 Intuitively, given a node $x$ in the tree decomposition, $\Pi$ is the partition of the agents in $G^x$ into coalitions such that each $\pi \in Pi$ contain some agent from $\beta(x)$. Moreover, to decide whether any agent that is forgotten and is a neighbor to some coalition, has a Nash deviation to some coalition in $\Pi$, we keep a indicator agent $F(a)$ for each agent that is part of $\Pi$ and we store the minimum utility this indicator agent has in its coalition.
 
Let $\Pi$ be  a coalition structure on a subset of agents $N'$ in  $G^x$. Recall We use $\cup\Pi$ to denote $N'$, we will use $\minU{\Pi}: \cup\Pi \rightarrow \Z$ to denote a mapping that stores for each agent $a \in \cup\Pi$, the utility of a neighbor $b \in N_{G^x}(a) \setminus \cup\Pi$ that has minimum utility among the neighbors of $a$ in $N_{G^x}(a) \setminus \cup\Pi$. That is, $b = \argmin_{b' \in N_{G^x}(a) \setminus \cup\Pi} u(b', \hat{\Pi})$ where $\hat{\Pi}$ is a partition of agents in $G^x$.
Additionally, to decide if any agent of $\Pi$ has a Nash deviation to some coalition of $G^x$ (possibly, forgotten), we will use $\maxU{\Pi}: \cup\Pi \rightarrow \Z$ to denote a mapping that stores the maximum utility that an agent $a \in N'$ would receive in a coalition $C$ in $G^x$. The coalition $C$ can be a coalition in $N_{G^x}(a) \setminus \cup\Pi$.
 Furthermore, we will use a boolean valued function $f: \Pi \rightarrow \{0,1\}$ that maps each coalition in $\Pi$ to $0$ or $1$, to indicate if more agents will be added to a coalition in $\Pi$. If for a coalition $\pi \in \Pi$, we have $f(\pi)=1$, then we say it is \emph{complete}.
We say $\Pi$ is \emph{partial Nash Stable} in $G^x$ if no agent in a complete coalition has a Nash deviation to another complete coalition in $G^x$.

We define the DP table \textsf{M}$[x, C, \topo(\extn{\Pi}), \minU{\Pi} , \maxU{\Pi},f]$ for every node $x$ of the tree decomposition, each partition $C$ of $\beta(x)$, each topology $\topo(\extn{\Pi})$ given $\beta(x)$ where $\Pi$ is a partition of at most $\sz\cdot \tw$ agents, and each functions $\minU{\Pi}$, $\maxU{\Pi}$ and $f$ where $\minU{\Pi}: \cup_{a\in \beta(x)} \Pi_a \rightarrow \Z$, $\maxU{\Pi}: \cup_{a\in \beta(x)} \Pi_a \rightarrow \Z$, and $f: \Pi \rightarrow \{0,1\}$. 
An entry of {\sf M} stores the maximum welfare in $G^x$ under the condition that the partition into coalitions satisfies the following properties.

\begin{enumerate}[label=P\arabic*]
\item\label{it:nashP1} \emph{$C$ and $\Pi$ are consistent}, i.e., the partition of the bag agents $\beta(x)$ in $G^x$ is denoted by $C$ and $\Pi$ is a partition satisfying, $C_a = \Pi_a \cap \beta(x)$ and $\Pi_a \subseteq \cup\Pi \subseteq V^x$ for each agent $a \in \beta(x)$. 
\item\label{it:nashP2} \emph{$\topo(\Pi)$ is consistent with $G^x$} i.e., the subgraph of $G^x$ induced on the agents in coalition of $a$ is $\topo(\Pi_a)$, i.e., $G^x[\Pi_a] = \topo(\Pi_a)$.

\item\label{it:nashP3} $\extn{\Pi}$ is consistent with $G^x$, i.e., for each agent $b \in \Pi_a$, we have that $F(b) = \emptyset$ if $N_{G^x}(b) \setminus \Pi_a = \emptyset$ for each $a \in \beta(x)$.

\item\label{it:nashP4} For each agent $a \in \beta(x)$ if $f(\Pi_a) = 0$, then more agents will be added to $\Pi_a$, otherwise no more agents will get added to the coalition $\Pi_a$ in future.

\item\label{it:nashPstable} $\Pi$ is partial Nash stable in $G^x$.
\end{enumerate}	
	 
Observe that we do not store $\Pi$. We only store the topology of $\extn{\Pi}$ which is a graph on at most $2\sz\cdot \tw$ agents. Given a partition $\Pi$, we will use $\Pi_{a\rightarrow\pi}$ to denote the partition formed due to deviation of an agent $a$ from its coalition $\Pi_a$ to another coalition $\pi \in \Pi \setminus \{\Pi_a\}$. 

	We say an entry of {\sf M}$[x,C, \topo(\extn{\Pi}),\minU{\Pi} , \maxU{\Pi},f]$ is \emph{valid} if the following conditions hold.
	
	 \begin{enumerate}[label=V\arabic*]
	\item\label{it:nashV1} \emph{$C$ and $\Pi$ are consistent}, i.e., $C_a = \Pi_a \cap \beta(x)$ for each agent $a\in \beta(x)$.
Moreover, either $\Pi_a = \Pi_b$, or $\Pi_a \cap \Pi_b = \emptyset$ for each pair of agents $a,b \in \beta(x)$ in $\topo(\extn{\Pi})$.
	
	\item\label{it:nashV2} \emph{$\topo(\extn{\Pi})$ is consistent with $G^x$ in $\beta(x)$}, i.e., for each pair of agents $a,b \in \beta(x)$ there is an edge $(a,b) \in \topo(\extn{\Pi})$ if and only if $(a,b)$ is an edge in $G^x$. Moreover, for each agent $a$ in $\Pi$, the unique edge incident to $F(a)$ is $(a, F(a))$.

	\item\label{it:nashV3} Let $x$ be a Introduce node whose bag contains one more agent, $a$, than its child. Then, $\minU{\Pi}(a) = L$ where $L$ is constant denoting the maximum welfare achieved any agent in a welfare maximizing partition of $G$ (since currently $a$ has no agent that is not in $\Pi$)  	and $\maxU{\Pi}(a) = u(a, \Pi)$ in the graph $\topo(\extn{\Pi})$.
	
	\item\label{it:deviate} For a pair of agents $a,b \in \extn{\Pi}$ such that $\Pi_a \neq \Pi_b$, if $f(\Pi_a) = f(\Pi_b) = 1$, then $u(a, \Pi) \geq u(a, \Pi_{a\rightarrow\Pi_b})$ and $u(b, \Pi) \geq u(b, \Pi_{b\rightarrow\Pi_a})$. Moreover, for each agent $a \in \extn{\Pi}$ such that $f(\Pi_a) = 1$, it holds that $u(a, \Pi) \geq 0$ in the graph $\topo(\extn{\Pi})$.
	
	\item\label{it:nashVstable} If $f(\Pi_a) = 1$ for some agent $a \in \beta(x)$, then for each agent $b$ in $\Pi$ it holds that $u(b, \Pi) \geq \maxU{\Pi}(b)$ and $\minU{\Pi}(b) \geq u(F(b), \Pi_{F(b) \rightarrow \Pi_a})$. That is, for each agent in $\Pi$ the maximum utility stored for an agent is at least its current utility and the stored utility of a friend must be at least the utility the friend would get by deviating to $\Pi_a$.
	\end{enumerate}
	
	The last two conditions of validity together imply that no agent in $G^x$ would prefer to deviate to a completed coalition that it is not part of.
	
	Once the table is computed correctly, the maximum welfare is given by the value stored in  {\sf M}$[r,C, \topo(\extn{\Pi}),\minU{\Pi} , \maxU{\Pi},f]$ where $C$ is empty partition and $\topo(\extn{\Pi})$ is empty.
	The basis corresponds to leaves (whose bags are empty), and are initialized to store $0$. For each entry that is not valid we store $-\infty$. 
	 To complete the proof, it now suffices to describe the computation of the table entries at each of the three non-trivial types of nodes in the decomposition and prove correctness.
	
	\noindent
	\textbf{Introduce Node.}
	Given a node $x$ whose bag contains one more agent, $a$, than its child $y$, 
	we recursively define {\sf M}$[x,C, \topo(\extn{\Pi}),\minU{\Pi} , \maxU{\Pi},f]$. 
	Let $C'$ be a partition of $\beta(y)$ such that  $C' =  C\res{\beta(y)}$, $\topo(\extn{\Pi'})$ denotes the graph we get by deleting agents $a$ and $F(a)$ from $\topo(\extn{\Pi})$, $\minU{\Pi'}$ is the function we get by removing $a$ from the domain of $\minU{\Pi}$, and $f'(\Pi_a\setminus\{a\}) = 0$ and $f'(\pi) = f(\pi)$ for each $\pi \in \Pi\setminus \{\Pi_a\}$. Let $S$ be the set of functions $\maxU{\Pi'}$ that satisfy the following:
\begin{align*}
\max \{ u(b, \Pi), \maxU{\Pi'}(b)\} = \maxU{\Pi}(b)  &\text{ ~~for each $b \in \Pi_a \setminus \{a\}$}\\
\maxU{\Pi'}(b) = \maxU{\Pi}(b) & \text{ ~~ otherwise}.
\end{align*}
	
 That is, we check if $\maxU{\Pi'}(b)$ is greater than the utility of $b$ in its current coalition $\Pi_a$  for each agent in coalition of $a$ and store the larger value in $\maxU{\Pi}(b)$. For the rest of agents the utility doesn't change, so, it remains the same. 	Let $\Delta(\topo(\extn{\Pi}), \topo(\extn{\Pi'}))$ denote the difference in welfare due to adding $a$ to the coalition $\Pi_a$. Then, we define $\Delta(\topo(\extn{\Pi}), \topo(\extn{\Pi'}))= \SW(\topo(\Pi_a)) -\SW(\topo(\Pi_a \setminus \{a\}))$ where $\topo(\Pi_a)$ is the graph induced on agents of $\Pi_a$ in the topology of $\extn{\Pi}$. Then,
	
	\begin{align*}
&{\sf M}[x,C, \topo(\extn{\Pi}), \minU{\Pi},\maxU{\Pi},f] \\
&=\max_{\maxU{\Pi'}\in S} \{{\sf M}[y,C', \topo(\extn{\Pi'}),\minU{\Pi'} , \maxU{\Pi'} ,f'] \\
&~~~~~~~~+ \Delta(\topo(extn{\Pi}), \topo(\extn{\Pi'})) \} \end{align*}

Since this is a valid entry, condition \eqref{it:nashPstable} holds, i.e., if $a$ is the last agent that is added to $\Pi_a$, then no agent would prefer to deviate from $\Pi_a$ and no agent in $G^x$ that is not in $\Pi$ would prefer to deviate to the coalition $\Pi_a$ after adding $a$ to it. The validity condition~\eqref{it:deviate} ensured that no agent in $\Pi$ whose coalition is complete prefer to deviate to $\Pi_a$.

\noindent
\textbf{Forget Node.}  Let $x$ be a node whose bag contains one fewer agent $a$ than its child $y$. 
We examine each coalition agent $a$ can be in (this the set of coalitions any of \(a\)'s neighbors is in), and take the one that maximizes the entry of the child. So, we consider any coalition structure $C'$ of $\beta(y)$ such that when $C'\res{\beta(x)} = C$. If $C'_a = \{a\}$, i.e., agent $a$ is in a different coalition in $\beta(y)$ than any agent in $\beta(x)$, then we take the maximum over all possible topologies, $\topo(\extn{\Pi'})$ where the coalition containing $a$ has, additionally, at most $\sz-1$ anonymous agents. 
Moreover, $\minU{\Pi'}$ and $\maxU{\Pi'}$ are two functions that satisfy the following: for each $b \in N(a) \cap \beta(x)$, it holds that $\minU{\Pi}(b) = \min\{ \minU{\Pi'}(b), \min_{a' \in \Pi'_a\cap N(b)} u(a', \Pi')\}$  and $\maxU{\Pi}(b) = \max\{\maxU{\Pi'}(b), u(b, \Pi'_{b\rightarrow\Pi'_a})\}$. 
The first condition says that we will consider only those functions $\minU{\Pi'}$ such that for each $b \in N(a) \cap \beta(x)$ the minimum utility of a neighbor of $b$ considering the forgotten agents of $\Pi'_a$ is the same as $\minU{\Pi}(b)$, and the later condition says that $\maxU{\Pi}$ must be a modification of  $\maxU{\Pi'}$ based on the maximum utility $b$ would receive by deviating to $\Pi_a$. 
Additionally, let $f'$ be defined as follows: $f'(\Pi'_a) = 1$ and $f'(\pi) = f(\pi)$ for each $\pi \in \Pi \setminus \{\Pi_a\}$.

Otherwise, $(\emptyset \neq) C'_a \setminus \{a\} \in C$. Then $\topo(\extn{\Pi'}) = \topo(\extn{\Pi})$, $\minU{\Pi'} = \minU{\Pi}$, and $\maxU{\Pi'} = \maxU{\Pi}$ . Additionally, function $f'$ is defined as follows: $f'(\Pi'_a) = 0$ and $f'(\pi) = f(\pi)$ for each $\pi \in \Pi \setminus \{\Pi_a\}$.
 Hence, we define,

\begin{align*}
&{\sf M}[x,C, \topo(\extn{\Pi}), \minU{\Pi},\maxU{\Pi},f]\\
&= \max_{C',\topo(\extn{\Pi'}),\minU{\Pi'} , \maxU{\Pi'} ,f'}{\sf M}[y,C', \topo(\extn{\Pi'}),\minU{\Pi'} , \maxU{\Pi'} ,f'] 
\end{align*}

\noindent
\textbf{Join node.} For a node $x$ with two children $y_1$ and $y_2$, we compute the sum of the entries associated with these children (where $C$ is the partition of coalitions in both $y_1$ and $y_2$), from which we modify the welfare of agents in $\beta(x)$ as they are computed in subgraphs $G^{y_1}$ and $G^{y_2}$ and not in $G^x$.

We define $\topo(\extn{\Pi^i})$, $\minU{\Pi^i}$, $\maxU{\Pi^i}$, And $f_i$ for each $i \in [2]$ to be the ones that satisfy the following conditions.
For each agent $a \in \beta(x)$, it holds  that  $\topo(\extn{\Pi^1})$ and $\topo(\extn{\Pi^2})$ restricted to the topology of $C_a$ are the same.  Let $\topo(\extn{\Pi^1}) \cup \topo(\extn{\Pi^2})$ define the topology we get by joining the two topologies at $\beta(x)$, i.e., by identifying the agents in $\beta(x)$ in the two topologies.

For each $i \in[2]$, and each agent $a \in \extn{\Pi^i} \setminus \beta(x)$,  it holds that  $\minU{\Pi^i}(a) = \minU{\Pi}(a)$ and for each agent $a\in \beta(x)$, it holds that $\minU{\Pi}(a) = \min\{\minU{\Pi^1}(a),\minU{\Pi^2}(a)\}$. For each agent $a \in \extn{\Pi^i} \setminus \beta(x)$, we have $\maxU{\Pi}(a) = \max\{\maxU{\Pi^i}(a), u (a, \Pi)\}$ for $i \in [2]$ and each agent $a \in \beta(x)$, it holds that $\maxU{\Pi}(a) = \max\{\maxU{\Pi^1}(a),\maxU{\Pi^2}(a), u (a, \Pi)\}$ where $u(a, \Pi)$ is computed in the graph $\topo(\extn{\Pi})$. 
Finally, if $\Pi^i_a = C_a$ that is all agents in the coalition of $a$ are in $\beta(x)$ in $\Pi$, $\Pi^1$, and $\Pi^2$, then $f_1 (\Pi^i_a)= f_2 (\Pi^i_a) = f(\Pi_a)$ for each $i \in [2]$. Otherwise, more agents get added to the coalition of $a$ in $\Pi$ than in $\Pi^1$ and $\Pi^2$. So $f_1(\Pi^i_a) = f_2(\Pi^i_a) = 0$ for each $a \in \beta(x)$.
Let $\Delta(\topo(\extn{\Pi}), \topo(\extn{\Pi^1}), \topo(\extn{\Pi^2}))$ denote the difference between the total utility of the coalitions that contain agents of $\beta(x)$ in $G^x$ and the total utility of these coalitions in $G^{y_1}$ and $G^{y_2}$. Recall that total utility of the agents in a partition $P$ is denoted by $\SW(P)$.
 That is $\Delta(\topo(\extn{\Pi}), \topo(\extn{\Pi^1}), \topo(\extn{\Pi^2})) = \SW(\topo(\Pi^1) \cup \topo(\Pi^2))- \SW(\topo(\Pi^1) -\SW(\topo(\Pi^2)$ where $\topo(\Pi^i)$ denotes the  induced subgraph of $\topo(\extn{\Pi^i})$ on the agents in $\Pi^i$ for each $i \in [2]$.

\begin{align*}
&{\sf M}[x,C,\topo(\extn{\Pi}), \minU{\Pi},\maxU{\Pi},f]\\
&= \max_{\substack{\topo(\extn{\Pi^i}),\minU{\Pi^i} , \maxU{\Pi^i} ,f_i:\\ i \in [2]}} \{{\sf M}[y_1,C,\topo(\extn{\Pi^1}), \minU{\Pi^!},\maxU{\Pi^1},f_1] \\
&~~~~~~~~~~~~~~~~~~~~~~~~~~~~~+  {\sf M}[y_2,C, \topo(\extn{\Pi^2}), \minU{\Pi^2},\maxU{\Pi^2},f_2] \\
& ~~~~~~~~~~~~~~~~~~~~~~~~~~~~~+ \Delta(\topo(\Pi),\!\topo(\Pi^1),\!\topo(\Pi^2)\!)\}. \end{align*}

Next we prove the correctness of each of the above formulas by showing inequality in both direction as above.
The argument to show that the entry ${\sf M}[x,C,\topo(\extn{\Pi}), \minU{\Pi},\maxU{\Pi},f]$ stores the maximum welfare is similar to Lemma~\ref{lem:welfare+coalDPcorrectness}. The stability of $\Pi$ follows from the validity checks.
We analyze the Introduce separately.  

Let $x$ be a leaf node. Since, $\beta(x)=\emptyset$, we only have one value ${\sf M}[x,C, \topo(\Pi)]  = 0$.

 \noindent
{\bf Introduce Node}: Let $x$ be a node whose bag contains one more agent, $a$, than its child $y$. Observe that,  $a$ is the new agent introduced in $C_a$ and hence, in $\Pi_a$. 
We only need to recompute the contribution of $\Pi_a$ in $\beta(x)$ towards welfare. Note that the rest of coalitions are in $C'$ as well and their contribution to welfare remains the same as ${\sf M}[y,C',\topo(\extn{\Pi'}), \minU{\Pi'},\maxU{\Pi'},f']$. Since we know the graph $\topo(\extn{\Pi})$ satisfy condition \eqref{it:nashV2} in the initial validity check, clearly, the sum of utility of agents in $\Pi_a$ in $G^x$ is $\SW(\topo(\Pi_a))$. Moreover,  since $\Pi\res{\beta(y)} = \Pi'$ and we get $\topo(\extn{\Pi'})$ by deleting $a$ from $\topo(\extn{\Pi})$, the total utility of agents in $\Pi_a \setminus \{a\}$ in $G^y$ is $\SW(\topo(\Pi_a\setminus\{a\})$. Therefore, the change in total utility is given by $\Delta(\topo(\extn{\Pi}), \topo(\extn{\Pi'}))$ since $G^y = G^x -a$.
Therefore, since ${\sf M}[y,C', \topo(\extn{\Pi'}), \minU{\Pi'},\maxU{\Pi'},f']$ is maximum welfare in $G^y$ given partition $\Pi'$, and $\minU{\Pi'}$ and $f'$ are the same as $\minU{\Pi}$ and $f$ restricted to $\beta(y)$, respectively, we have that ${\sf M}[x,C,\topo(\extn{\Pi}), \minU{\Pi},\maxU{\Pi},f]$ is at most ${\sf M}[y,C', \topo(\extn{\Pi'}), \minU{\Pi'},\maxU{\Pi'},f']+ \Delta(\topo(\extn{\Pi}), \topo(\extn{\Pi'}))$. 
Therefore, ${\sf M}[x,C,\topo(\extn{\Pi}), \minU{\Pi},\maxU{\Pi},f]$ is at most $\max_{\maxU{\Pi'}\in S} \{{\sf M}[y,C', \topo(\extn{\Pi'}), \minU{\Pi'},\maxU{\Pi'},f']+ \Delta(\topo(\extn{\Pi}), \topo(\extn{\Pi'}))\}$.

Now since ${\sf M}[x,C,\topo(\extn{\Pi}), \minU{\Pi},\maxU{\Pi},f]$ stores the maximum welfare for $\Pi$ in $G^x$ and $G^y$ is a subgraph of $G^x$, we have that any partition in $G^y$ can have total utility at most ${\sf M}[x,C,\topo(\extn{\Pi}), \minU{\Pi},\maxU{\Pi},f]$ barring the change in total utility of agents due to removing $a$ from its coalition in $\Pi$. Let $\Pi^y  = \Pi\res{\beta(y)}$ and $f'$ is as defined in the recurrence. Then,
${\sf M}[x,C,\topo(\extn{\Pi}), \minU{\Pi},\maxU{\Pi},f]$ is at least  ${\sf M}[y,C', \topo(\extn{\Pi^y}), \minU{\Pi^y},\maxU{\Pi^y},f']+ \Delta(\topo(\extn{\Pi}), \topo(\extn{\Pi^y})) $.
  Additionally, since $\Pi$ maximizes welfare in $G^x$, the right hand side of the equation is maximum when the function $\maxU{\Pi^y}$ is considered.
 Hence, we proved the formula.

Similarly for Forget and Join node it can be check that  the modification of $\minU{\Pi}, \maxU{\Pi}$, and $f$ are done according to the change in $\Pi$ and $\topo(\extn{\Pi})$. Then the proof for welfare maximization for Forget and Join node is the same as in Lemma~\ref{lem:welfare+coalDPcorrectness}.
This concludes the description and the proof of correctness of the recursive formulas for computing the values of the table ${\sf M}$.

Therefore it remains to prove that each entry of {\sf M} either stores $-\infty$ or satisfy \eqref{it:nashP1}-\eqref{it:nashPstable}.  We check these conditions in the validity check. 
Condition \eqref{it:nashP1}, \eqref{it:nashP2}, and \eqref{it:nashP3} holds true due to the valdity checks \eqref{it:nashV1}, \eqref{it:nashV2}, and \eqref{it:nashV3}, respectively. Condition \eqref{it:nashP4} holds since in our computation when an agent (Introduce node) or agents (Join node) are added to a coalition, we check if $f(C_a) =0$ for the coalition $C_a$ that the new agent $a$ is joining in both Introduce and Join node. Finally, \eqref{it:nashPstable} holds due to conditions \eqref{it:deviate} and \eqref{it:nashVstable}. The function $\maxU{\Pi}$ stores the maximum utility an agent can receive from any complete coalitions in $G^x$. We check that whenever an agent $a$ joins the coalition $\Pi_a$, no agent $b$ in the coalition $\Pi_a$ wants to deviate from $\Pi_a$ by checking their maximum possible utility stored in $\maxU{\Pi}(b)$ is at least the current utility of $b$. Moreover, we check that no agent $b'$ in $G^x$ want to join $\Pi_a$. It is sufficient to check this condition for the agents that are neighbor to some agent in $\Pi_a$ since each coalition must be connected. In $\extn{\Pi}$ we have stored a neighbor for each agent in $b \in\Pi_a$ denoted by $F(b)$ and $\minU{\Pi}(b)$ stores the utility of $F(b)$ that has minimum utility. Hence, if there exists an agent $b'$ in the neighborhood of $b$ such that $b'$ wants to deviate to $\Pi_a$, then we  identify such an agent by checking if the $\minU{\Pi}(b)$ is at least the utility of $F(b)$ when deviates to $\Pi_a$. We can calculate the utility of $F(b)$ after deviation since we know the topology $\extn{\Pi}$. Therefore, no agent wants to deviate to or from $\Pi_a$. Hence, there is no Nash deviation in $G^x$.
This completes the proof of correctness of the DP. 

For running time of the DP, observe that each validity check can be done in polynomial time. The table {\sf M} has  $(\tw \cdot 2^{\Oh{\sz^2}})^{\Oh{\tw}}$ entries and each entry can be computed in time  $(\tw \cdot 2^{\Oh{\sz^2}})^{\Oh{\tw}}$. Hence, we prove the theorem.

\end{proof}
\fi

From \Cref{lem:degen_coal_size} it follows that  if $\val(2) < 0$ and $\tw(G)$ is bounded, then the maximum coalition size  of an welfare maximizing outcome is bounded. Hence, using Theorem~\ref{thm:tw+coal-sz}  we get the following.
\begin{corollary}
	\DGNS, \DGIR, and \DGWF are fixed-parameter tractable parameterized by the treewidth $\tw(G)$ if $\val(2) < 0$.
\end{corollary}

Turning back to general scoring vectors, we recall that \Cref{lem:maxdeg_coal_size} provided a bound on the size of the coalitions in a welfare-optimal outcome in terms of the maximum degree $\Delta(G)$ of the network $G$. Applying Theorem~\ref{thm:tw+coal-sz} again yields:

\begin{corollary}
	\DGNS, \DGIR, and \DGWF are fixed-parameter tractable parameterized by the treewidth $\tw(G)$ and the maximum degree $\Delta(G)$ of the social network.
\end{corollary}

As our final contribution, we provide fixed-parameter algorithms for computing welfare-optimal outcomes that can also deal with networks containing high-degree agents. To do so, we exploit a different structural parameter than the treewidth---notably the vertex cover number of $G$ ($\vc(G)$). We note that while the vertex cover number is a significantly more ``restrictive'' graph parameter than treewidth, it has found numerous applications in the design of fixed-parameter algorithms, including for other types of coalition games~\cite{BiloFMM18,BodlaenderHJOOZ20,HanakaL22}.

\iflong
\begin{theorem}
\fi
\ifshort
\begin{theorem}[$\star$]
\fi
\label{thm:WfIsNsFPTwrtVC}
	\DGNS, \DGIR, and \DGWF are fixed-parameter tractable parameterized by the vertex cover number $\vc(G)$ of the social network.
\end{theorem}
\ifshort
\begin{proof}[Proof Sketch]
Let $k = \vc(G)$ and let $U$ be a vertex cover for~$G$ of of size~$k$.
	Observe that in each solution there are at most $k$ non-singleton coalitions, since~$G$ has a vertex cover of size~$k$ and each coalition must be connected.
	Furthermore, the vertices of $G - U$ can be partitioned into at most $2^k$ groups according to their neighborhood in the set~$U$.
	That is, there are $n_W$ vertices in $G - U$ such that their neighborhood is $W$ for some $W \subseteq U$; denote this set of vertices~$I_W$.
	
We performing exhaustive branching to determine certain information about the structure of the coalitions in a solution---notably:
	\begin{enumerate}
		\item
			which vertices of $U$ belong to each coalition (i.e., we partition the set $U$); note that there are at most $k^k$ such partitions, and
		\item
			if there is at least one agent of $I_W$ in the coalition or not (for each $W$ a subset of the guessed subset of~$U$ for the coalition); note that there are at most $(2^{2^k})^k$ such assignments of these sets to the coalitions.
	\end{enumerate}
	We branch over all possible admissible options of the coalitional structure described above possessed by a hypothetical solution. The total number of branches is upper-bounded by a function of the parameter value~$k$ and thus for the problems to be in \FPT it suffices to show that for each branch we can find a solution (if it exists) by a fixed-parameter subprocedure.
	To conclude the proof, we show that a welfare-maximum outcome (which furthermore satisfies the imposed stability constraints) with a given coalitional structure can be computed by modeling this as an Integer Quadratic Program where $d+\|A\|_{\infty}+ \|Q\|_{\infty}$ are all upper-bounded by a function of $k$, allowing us to invoke Proposition~\ref{prop:IQPisFPT}.
	
	The (integer) variables of the model are $x^C_W$, which express the number of vertices from the set $I_W$ in the coalition with $C \subseteq U$; thus, we have $x^C_W \in \mathbb{Z}$ and $x^C_W \ge 1$.
	Let $\mathcal{C}$ be the considered partitioning of the vertex cover~$U$.
	We use $C \in \mathcal{C}$ for the set $C \subseteq U$ in the coalition and $C^+$ for the set $C$ and the guessed groups having at least one agent in the coalition.
	We require that the vertices of $G-U$ are also partitioned in the solution, i.e.,
	\begin{equation}\label{eq:WfIsNsFPTwrtVC:IQP:partition}
		\sum_{C \in \mathcal{C}} \sum_{W \in C^+} x^C_W = n_W \qquad \forall W \subseteq U.
	\end{equation}
	The quadratic objective expresses the welfare of the coalitions in the solution while the linear constraints ensure the stability of the outcome; for the latter, we rely on the fact that it is sufficient to verify the stability for a single agent from the group~$I_W$ in each coalition.
\end{proof}
\fi

\iflong
\begin{proof}
	At its core, the algorithm combines a branching procedure with an IQP formulation.
		Let $k = \vc(G)$ and let $U$ be a vertex cover for~$G$ of of size~$k$.
	Note that in each solution there are at most $k$ non-singleton coalitions, since~$G$ has a vertex cover of size~$k$ and each coalition must be connected.
	Furthermore, the vertices of $G - U$ can be partitioned into at most $2^k$ groups according to their neighborhood in the set~$U$.
	That is, there are $n_W$ vertices in $G - U$ such that their neighborhood is $W$ for some $W \subseteq U$; denote this set of vertices~$I_W$.

	As for the coalitional structure we guess:
	\begin{enumerate}
		\item
			which vertices of $U$ belong to each coalition (i.e., we partition the set $U$); note that there are at most $k^k$ such partitions, and
		\item
			if there is at least one agent of $I_W$ in the coalition or not (for each $W$ a subset of the guessed subset of~$U$ for the coalition); note that there are at most $(2^{2^k})^k$ such assignments of these sets to the coalitions.
	\end{enumerate}
	We try all possible admissible guesses of the coalitional structure of the solution.
	All in all the number of guesses is a function of the parameter value~$k$ and thus for the problems to be in \FPT it suffices to show that for a fixed guess we can find a solution (if it exists) in \FPT-time.
	For each guess we find a solution that maximizes the welfare for such a structure subject to the desired stability using an IQP model.
	The (integer) variables of the model are $x^C_W$ expressing the number of vertices from the set $I_W$ in the coalition with $C \subseteq U$; thus, we have $x^C_W \in \mathbb{Z}$ and $x^C_W \ge 1$.
	Let $\mathcal{C}$ be the guessed partition of the vertex cover~$U$.
	We use $C \in \mathcal{C}$ for the set $C \subseteq U$ in the coalition and $C^+$ for the set $C$ and the guessed groups having at least one agent in the coalition.
	We require that the vertices of $G-U$ are also partitioned in the solution, i.e.,
	\begin{equation}\label{eq:WfIsNsFPTwrtVC:IQP:partition}
		\sum_{C \in \mathcal{C}} \sum_{W \in C^+} x^C_W = n_W \qquad \forall W \subseteq U.
	\end{equation}
	The quadratic objective expresses the welfare of the coalitions in the solution while the linear constraints ensure the stability.
	Here, for the stability we rely on the fact that it is sufficient to verify the stability for a single agent from the group~$I_W$ in the guessed coalition.

	As for the IQP model we first observe that all the distances in~$G$ are at most $2k$.
	Furthermore, for each pair of vertices in a coalition their distance can be determined already from the guess of the coalitional structure; denote this distance $\dist_{C^+}(u,v)$.
	We extend the notion of distance to sets $I_W$ and $I_{\widetilde{W}}$ with $W, \widetilde{W} \subseteq C^+$, that is, $\dist_{C^+}(W, \widetilde{W}) = \dist_{C^+}(u,v)$ for some $u \in I_W$ and $v \in I_{\widetilde{W}}$.
	Now, the objective is
	\[
		\max \sum_{C \in \mathcal{C}} \sum_{W, \widetilde{W} \in C^+} x^C_W \cdot \val\mathopen{}\left(\dist_{C^+}(W, \widetilde{W})\right) \cdot x^C_{\widetilde{W}}
	\]
	which is clearly a quadratic form in the variables of the IQP.
	This finishes the description of the IQP for \DGWF{} as we can check if at least one guess yields a solution with welfare at least the desired bound.
	The IQP can be solved in \FPT-time with respect to~$k$ as all the numbers in the matrices $A$ and $Q$ are constants and the dimension is at most~$k \cdot 2^k$.

	\noindent\textbf{Stability.}
	The most important notion for including the stability of the solution into our IQP is the utility of an agent.
	Note that we need to determine utilities only for agents $v \in I_W$ with respect to a coalition $C^+$ with $W \in C^+$.
	The utility of such an agent $v$ is then
	\(
		\operatorname{util}_{C}(v)
		=
		\val(2) \cdot (x^C_W - 1) + \sum_{\widetilde{W} \in C^+, \widetilde{W} \neq W} \val\mathopen{}\left(\dist_{C^+}(W, \widetilde{W})\right) x^C_{\widetilde{W}} + \sum_{u \in C} \val\mathopen{}\left(\dist_{C^+}(W, u)\right) \,.
	\)
	Here, the first summand is the utility~$v$ receives from the other vertices in its group (note that these are always in distance exactly~\(2\)).
	The second summand is the utility~$v$ receives from other vertices of $G-U$ assigned to its coalition and the last summand is the utility of~$v$ gained from the vertices from the set~$U$ in its coalition.
	Note that $\operatorname{util}_{C}(v)$ is a linear function in the integer variables~$x$.
	To that end, requiring
	\begin{equation}
		\operatorname{util}_{C}(v) \ge 0 \quad \forall C \in \mathcal{C}, \forall\, W \in C^+, \text{for some } v \in I_W
	\end{equation}
	ensures the individual rationality of the solution.
	
	For the Nash stability we note that the utility of $v$ with respect to a coalition $C$ is either negative (if $C \cap N(v) = \emptyset$) or it can be defined as above (as if $v$ belongs to $C^+$).
	To that end, we (further) require
	\begin{equation}
		\operatorname{util}_{C}(v) \ge \operatorname{util}_{\widetilde{C}}(v) \quad \forall C, \widetilde{C} \in \mathcal{C}, \forall\, W \in C^+\,,
	\end{equation}
	where $v$ is some agent in $I_W$.

	Note that the models for individual rationality and Nash stability only add constraints to the model.
	These new constraints only rely on constant coefficients (the values of the $\val$ vector).
	Thus one can solve the IQPs in \FPT-time with respect to the parameter value~$k$.

	Now, it is not hard to see that if the given instance admits a stable solution, then if we set the values of variables in our IQP according to numbers of vertices from $I_W$ in that coalition, this forms a feasible solution.
	Furthermore, if any of the above IQPs admits a solution, then there is a (stable) solution of the original instance with the social welfare determined by the objective function.
	Therefore, the presented models are valid.
\end{proof}
\fi

\section{Conclusions and Future Research Directions}

In this work, we studied social distance games through the lens of an adaptable, non-normalized scoring vector which can capture the positive as well as negative dynamics of social interactions within coalitions.
The main focus of this work was on welfare maximization, possibly in combination with individual-based stability notions---individual rationality and Nash stability.
It is not surprising that these problems are intractable for general networks; we complement our model with algorithms that work well in tree-like environments.

Our work opens up a number of avenues for future research.
One can consider other notions of individual-based stability such as individual stability~\cite{BrandtCELP16,GanianHKSS22}, or various notions of group-based stability such as core stability~\cite{BranzeiL2011,BrandtCELP16,OhtaBISY17}.
Furthermore, our results do not settle the complexity of finding stable solutions (without simultaneous welfare maximization).
Therefore, it remains open if one can find a Nash stable solution for a specific scoring vector.
Also, a more complex open problem is to characterize those scoring vectors that guarantee the existence of a Nash (or individually) stable solution.

Finally, we remark that the proposed score-based \DGs\ model can be generalized further, e.g., by allowing for a broader definition of the scoring vectors. For instance, one could consider situations where the presence of an agent that is ``far away'' does not immediately set the utility of other agents in the coalition to $-\infty$. One way to model these settings would be to consider ``\emph{open}'' scoring vectors, for which we set $\val(a)=\val(\cdiam)$ for all $a>\cdiam$---meaning that distances over $\cdiam$ are all treated uniformly but not necessarily as unacceptable. 

Notice that if $\val(\cdiam) \geq 0$ for an open scoring vector $\val$, the grand coalition is always a social-welfare maximizing outcome for all three problems---hence here it is natural to focus on choices of $\val$ with at least one negative entry. We note that all of our fixed-parameter algorithms immediately carry over to this setting for arbitrary choices of open scoring vectors $\val$. The situation becomes more interesting when considering the small-world property: while the diameter of every welfare-maximizing outcome can be bounded in the case of Nash stable or individually rational coalitions (as we prove in our final Theorem~\ref{thm:diam-nashir} below), whether the same holds in the case of merely trying to maximize social welfare is open and seems to be a non-trivial question. Because of this, Theorem~\ref{thm:twstable} can also be extended to the setting with open scoring vectors, but it is non-obvious whether Theorem~\ref{thm:tw} can.

\ifshort
\begin{theorem}[$\star$]
\fi
\iflong
\begin{theorem}
\fi
\label{thm:diam-nashir}
	Let $\val=(\maxval,\dots,\cdiam)$ be an arbitrary open scoring vector and $G$ be a social network. Every outcome $\Pi$ containing a coalition $C\in\Pi$ with diameter exceeding \(\ell = 2\maxval \cdiam\) can be neither Nash-stable nor individually rational.
\end{theorem}
\ifshort
\begin{proof}[Proof Sketch]
Consider a shortest path $P$ in $C$ whose length exceeds $\ell$. We identify a set of edge cuts along $P$ and show that at least one such cut must be near an agent whose utility in $C$ is negative, due to the presence of a large number of agents that must be distant from the chosen edge cut. 
\end{proof}
\fi

\iflong
\begin{proof}
	Assume for contradiction that there is a coalition \(C\) with diameter at least \(\ell\) and consider \(s\) and \(t\) at maximum distance in \(C\), and for \(i \in [\ell] \cup \{0\}\) \(B_i\) as the set of agents at distance at most \(i\) from \(s\) in \(C\).
	We consider all \(i\) that are multiples of \(2\cdiam\), and the edge cuts between \(B_i\) and \(B_\ell \setminus B_i\) at these \(i\).
	This defines \(\frac{\ell}{2\cdiam} = \maxval + 1\)-many cuts which we denote as \(F_{2\cdiam j}\) for \(j \in [\maxval] \cup \{0\}\).
	
	For an edge cut in \(C\), we say agents at distance more than \(\cdiam\) from the endpoints of cut edges are \emph{far} from the cut and agents at distance at most \(\cdiam\) from the endpoints of cut edges are \emph{close} to the cut.
	This means that the sets of agents close to the edge cuts \(F_{2\cdiam j}\) as described above are pairwise disjoint.
	Hence any such cut \(F_{2\cdiam j^*}\) minimizing the number of agents close to it has at most \(\frac{n}{\maxval + 1}\) agents close to it, and correspondingly at least \(\maxval\) as many agents far from it.
	
	Now consider an agent that is close to \(F_{2\cdiam j^*}\).
	This agent values \(C\) negatively because it receives negative utility for each agent far from \(F_{2 \cdiam j^*}\), and at most \(\maxval\) positive utility for each agent close to \(F_{2 \cdiam j^*}\) other than itself, i.e. its utility is at most \(-\maxval \cdot z + \maxval (z - 1) < 0\) where \(z\) is the number of agents close to \(F_{2 \cdiam j^*}\).
	Hence the agent would prefer to be in a singleton coalition alone, which contradicts individual rationality as well as Nash stability.
	\end{proof}
\fi

\paragraph*{Acknowledgements.} 
All authors are grateful for support from the OeAD bilateral Czech-Austrian WTZ-funding Programme (Projects No. CZ 05/2021 and 8J21AT021). Robert Ganian acknowledges support from the Austrian Science Foundation (FWF, project Y1329). Thekla Hamm also acknowledges support from FWF, project J4651-N. Dušan Knop, Šimon Schierreich, and Ondřej Suchý acknowledge the support of the Czech Science Foundation Grant No. 22-19557S. Šimon Schierreich was additionally supported by the Grant Agency of the Czech Technical University in Prague, grant \mbox{No.~SGS23/205/OHK3/3T/18}.

\bibliographystyle{named}
\bibliography{references}

\end{document}